\RequirePackage{fix-cm}
\documentclass[final]{svjour3}                     
\smartqed  

\usepackage{epsfig}
\usepackage{amsmath}
\usepackage{amssymb}
\usepackage{setspace}
\usepackage{enumerate}
\usepackage[ruled,linesnumbered,vlined]{algorithm2e}
\usepackage{color}
\usepackage{array}
\newcolumntype{C}[1]{>{\centering\let\newline\\\arraybackslash\hspace{0pt}}m{#1}}

\newcounter{exctr}

\newtheorem{defn}{\noindent $\mathbf{Definition}$}[section]
\numberwithin{defn}{section}

\newtheorem{thm}[defn]{$\mathbf{Theorem}$}

\begin{document}

\title{A Linear Formulation for Disk Conformal Parameterization of Simply-Connected Open Surfaces}

\titlerunning{Linear Disk Conformal Parameterization}        

\author{Gary Pui-Tung Choi \and Lok Ming Lui}

\institute{Gary Pui-Tung Choi \at
              John A. Paulson School of Engineering and Applied Sciences, Harvard University\\
              \email{pchoi@g.harvard.edu}
            \and
            Lok Ming Lui \at
               Department of Mathematics, The Chinese University of Hong Kong\\
              \email{lmlui@math.cuhk.edu.hk}
}

\maketitle

\begin{abstract}
Surface parameterization is widely used in computer graphics and geometry processing. It simplifies challenging tasks such as surface registrations, morphing, remeshing and texture mapping. In this paper, we present an efficient algorithm for computing the disk conformal parameterization of simply-connected open surfaces. A double covering technique is used to turn a simply-connected open surface into a genus-0 closed surface, and then a fast algorithm for parameterization of genus-0 closed surfaces can be applied. The symmetry of the double covered surface preserves the efficiency of the computation. A planar parameterization can then be obtained with the aid of a M\"obius transformation and the stereographic projection. After that, a normalization step is applied to guarantee the circular boundary. Finally, we achieve a bijective disk conformal parameterization by a composition of quasi-conformal mappings. Experimental results demonstrate a significant improvement in the computational time by over 60\%. At the same time, our proposed method retains comparable accuracy, bijectivity and robustness when compared with the state-of-the-art approaches. Applications to texture mapping are presented for illustrating the effectiveness of our proposed algorithm.

\end{abstract}

\section{Introduction}
With the advancement of computer technologies and 3D acquisition techniques, 3D objects nowadays are usually captured and modeled by triangular meshes for further usages. A large number of applications of triangular meshes can be found in computer graphics and computer-aided design. However, working on general meshes is a difficult task because of their complicated geometry. The complicated geometry hinders applications such as surface registration, morphing and texture mapping. To overcome this problem, one common approach is to parameterize the surfaces onto a simple parameter domain so as to simplify the computations. For instance, textures can be designed on the simple domain and then mapped back onto the original surfaces \cite{Haker00,Levy02,Zhang05}. Another example that usually makes use of parameterization is surface registration \cite{Gu04,Lui10,Choi15}. Instead of directly computing the registration between two convoluted surfaces, one can perform the registration on the simple parameter domain, which is much easier. It is also common to perform surface remeshing with the aid of parameterizations \cite{Hormann01,Praun03,Remacle10}. With the development of the computer industry, the problem of finding a good parameterization method is becoming increasingly important.

To make a parameterization useful and applicable, one should seek for a method that minimizes certain types of distortions. In particular, it is desirable to minimize the angular distortions of the 3D meshes. {\it Angle preserving parameterizations}, also known as {\it conformal parameterizations}, effectively preserve the local geometry of the surfaces. Hence, in this paper, our goal is to develop an efficient conformal parameterization algorithm.

The choice of the parameter domain is also a key factor in deciding the parameterization scheme. For simply-connected open surfaces, one popular choice of the parameter domain is the unit disk. Using the unit disk as a parameter domain is advantageous in the following two aspects. Firstly, the existence of the conformal parameterization is theoretically guaranteed. By the uniformization theorem, every simply-connected open surface is conformally equivalent to the open unit disk. Secondly, unlike free and irregular shapes on the plane, a consistent circular boundary of the parameter domain facilitates the comparisons and mappings between different surfaces.

In real applications, besides the quality of the parameterization result, it is also important to consider the computational efficiency of the parameterization algorithm. In particular, a fast algorithm is desired so that the computation can be completed within a short time. In this paper, we develop an efficient algorithm for the disk conformal parameterization of simply-connected open surfaces. To achieve the efficiency, we first transform a topologically disk-like surface to a genus-0 closed surface by a double covering technique. Then we can apply a fast parameterization algorithm for genus-0 closed surfaces to compute a spherical parameterization of the double covered surface. Note that although the size of the problem is doubled by double covering, the computational efficiency is preserved because of the symmetry of the double covered surface. The spherical parameterization, together with a suitable M\"{o}bius transformation and the stereographic projection, provides us with an almost circular planar parameterization for the original surface. A normalization technique followed by a composition of quasi-conformal maps are then used for obtaining a bijective disk conformal parameterization. The bijectivity of the parameterization is supported by quasi-conformal theories. The entire algorithm only involves solving sparse linear systems and hence the computation of the disk conformal parameterization is greatly accelerated. Specifically, our proposed method speeds up the computation of disk conformal parameterizations by over 60\% while attaining accuracy comparable to the state-of-the-art approaches \cite{Choi15b,Gu02}. In addition, our proposed method demonstrates robustness to highly irregular triangulations.

The rest of the paper is organized as follows. In Section \ref{previous}, we review the previous works on surface parameterizations. In Section \ref{contributions}, we highlight the contribution of our work. Our proposed algorithm is then explained in details in Section \ref{main}. The numerical implementation of the algorithm is introduced in Section \ref{implementation}. In Section \ref{experiments}, we present numerical experiments to demonstrate the effectiveness of our proposed method. The paper is concluded in Section \ref{conclusion}.

\section{Previous works} \label{previous}
\begin{table}[t]
    \centering
    \begin{tabular}{ |C{45mm}|c|c|c|c| }
    \hline
    Methods & Boundary & Bijective? & Iterative?\\ \hline
    Shape-preserving \cite{Floater97}  & Fixed & Yes & No \\ \hline
    MIPS \cite{Hormann00} & Free & Yes & Yes \\ \hline
    ABF/ABF++ \cite{Sheffer00,Sheffer05} & Free & Local (no flips) & Yes \\ \hline
    LSCM/DNCP \cite{Desbrun02,Levy02} & Free & No & No\\ \hline
    Holomorphic 1-form \cite{Gu02} & Fixed & No & No \\ \hline
    Mean-value \cite{Floater03} & Fixed & Yes & No\\ \hline
    Yamabe Riemann map \cite{Luo04} & Fixed & Yes & Yes \\ \hline
    Circle patterns \cite{Kharevych05} & Free & Local (no flips) & Yes \\ \hline
    Genus-0 surface conformal map \cite{Jin05} & Free & No & Yes \\ \hline
    Discrete Ricci flow \cite{Jin08}& Fixed & Yes & Yes \\ \hline
    Spectral conformal \cite{Mullen08} & Free & No & No \\ \hline
    Generalized Ricci flow \cite{Yang09} & Fixed & Yes & Yes \\ \hline
    Two-step iteration \cite{Choi15b} & Fixed & Yes & Yes \\ \hline
    \end{tabular}
    \caption{Several previous works on conformal parameterization of simply-connected open surfaces.}
    \label{previouswork}
\end{table}

With a large variety of real applications, surface parameterization has been extensively studied by different research groups. Readers are referred to  \cite{Floater02,Floater05,Hormann07,Sheffer06} for surveys of mesh parameterization methods. In this section, we give an overview of the works on conformal parameterization.

A practical parameterization scheme should retain the original geometric information of a surface as far as possible. Ideally, the isometric parameterization, which preserves geometric distances, is the best parameterization in the sense of geometry preserving. However, isometric planar parameterizations only exist for surfaces with zero Gaussian curvature \cite{Docarmo76}. Hence, it is impossible to achieve isometric parameterizations for general surfaces. A similar yet far more practical substitute is the conformal parameterization. Conformal parameterizations are angle preserving, and hence the infinitesimal shape is well retained. For this reason, numerous studies have been devoted to surface conformal parameterizations.

The existing algorithms for the conformal parameterizations of disk-type surfaces can be divided into two groups, namely, the free-boundary methods and the fixed-boundary methods. For the free-boundary methods, the planar conformal parameterization results are with irregular shapes. In \cite{Hormann00}, Hormann and Greiner proposed the MIPS algorithm for conformal parameterizations of topologically disk-like surfaces. The boundary develops naturally with the algorithm. In \cite{Sheffer00}, Sheffer and de Sturler proposed the Angle Based Flattening (ABF) method to compute conformal maps, based on minimizing the angular distortion in each face to the necessary and sufficient condition of a valid 2D mesh. Sheffer {\it et al.} \cite{Sheffer05} extended the ABF method to ABF++, a more efficient and robust algorithm for planar conformal parameterizations. A new numerical solution technique, a new reconstruction scheme and a hierarchical technique are used to improve the performance.  L\'{e}vy {\it et al.} \cite{Levy02} proposed the Least-Square Conformal Maps (LSCM) to compute a conformal parameterization by approximating the Cauchy-Riemann equation using the least square method. In \cite{Desbrun02}, Desbrun {\it et al.} proposed the Discrete, Natural Conformal Parameterization (DNCP) by computing the discrete Dirichlet energy. In \cite{Kharevych05}, Kharevych {\it et al.} constructed a conformal parameterization based on circle patterns, which are arrangements of circles on every face with prescribed intersection angles. Jin {\it et al.} \cite{Jin05} applied a double covering technique \cite{Gu03} and an iterative scheme for genus-0 surface conformal mapping in \cite{Gu02} to obtain a planar conformal parameterization. In \cite{Mullen08}, Mullen {\it et al.} reported a spectral approach to discrete conformal parameterizations, which involves solving a sparse symmetric generalized eigenvalue problem.

When compared with the free-boundary approaches, the fixed-boundary approaches are advantageous in guaranteeing a more regular and visually appealing silhouette. In particular, it is common to use the unit circle as the boundary for the conformal parameterizations of disk-type surfaces. Numerous researchers have proposed brilliant algorithms for disk conformal parameterizations. Floater \cite{Floater97} proposed the shape-preserving parameterization method for surface triangulations by solving linear systems based on convex combinations. In \cite{Floater03}, Floater improved the parameterization method using a generalization of barycentric coordinates. In \cite{Gu02}, Gu and Yau constructed a basis of holomorphic 1-forms to compute conformal parameterizations. By integrating the holomorphic 1-forms on a mesh, a globally conformal parameterization can be obtained. In \cite{Luo04}, Luo proposed the combinatorial Yamabe flow on the space of all piecewise flat metrics associated with a triangulated surface for the parameterization. In \cite{Jin08}, Jin \emph{et al.} suggested the discrete Ricci flow method for conformal parameterizations, based on a variational framework and circle packing. Yang \emph{et al.} \cite{Yang09} generalized the discrete Ricci flow and improved the computation by allowing two circles to intersect or separate from each other, unlike the conventional circle packing-based method \cite{Jin08}. In \cite{Choi15b}, Choi and Lui presented a two-step iterative scheme to correct the conformality distortions at different regions of the unit disk. Table \ref{previouswork} compares several previous works on the conformal parameterizations of disk-type surfaces.

Our proposed algorithm involves a step of spherical parameterization. Various spherical parameterization algorithms have been developed in the recent few decades, such as \cite{Angenent99,Haker00,Gotsman03,Gu02,Gu03,Gu04,Choi15}. Among the existing algorithms, we apply the fast spherical conformal parameterization algorithm proposed in \cite{Choi15}. More details will be explained in Section \ref{main}.

\section{Contributions} \label{contributions}
In this paper, we introduce a linear formulation for the disk conformal parameterization of simply-connected open surfaces. Unlike the conventional approaches, we first find an initial map via a parameterization algorithm for genus-0 closed surfaces, with the aid of a double covering technique. The symmetry of the double covered surface helps retaining the low computational cost of the problem. After that, we normalize the boundary and apply quasi-conformal theories to ensure the bijectivity and conformality. Our proposed algorithm is advantageous in the following aspects:
\begin{enumerate}
 \item Our proposed method only involves solving a few sparse symmetric positive definite linear systems of equations. It further accelerates the computation of disk conformal parameterizations by over 60\% when compared with the fastest state-of-the-art approach \cite{Choi15b}.
 \item With the significant improvement of the computational time, our proposed method possesses comparable accuracy as of the other state-of-the-art approaches.
 \item The bijectivity of the parameterization is supported by quasi-conformal theories. No foldings or overlaps exist in the parameterization results.
 \item Our proposed method is highly robust to irregular triangulations. It can handle meshes with very sharp and irregular triangular faces.
\end{enumerate}

\section{Proposed method} \label{main}
In this section, we present our proposed method for disk conformal parameterizations of simply-connected open surfaces in details.

A map $f:M \to N$ between two Riemann surfaces is called a {\it conformal map} if it preserves the inner product between vectors in parameter space and their images
in the tangent plane of the surface, up to a scaling factor. More specifically, there exists a positive scalar function $\lambda$ such that $f^* ds^2_N = \lambda ds^2_M$. In other words, conformal maps are angle preserving. The following theorem guarantees the existence of several special types of conformal maps.

\begin{thm} [Uniformization Theorem]
Every simply connected Riemann surface is conformally equivalent to exactly one of the following three domains:
\begin{enumerate}[(i)]
 \item the Riemann sphere,
 \item the complex plane,
 \item the open unit disk.
\end{enumerate}
\end{thm}
\begin{proof}
See \cite{Schoen94}. \qed
\end{proof}

With this theoretical guarantee, our goal is to efficiently and accurately compute a conformal map $f:M \to \mathbb{D}$ from a topologically disk-like surface $M$ to the open unit disk $\mathbb{D}$.

\begin{table}[t]
    \centering
    \begin{tabular}{ |C{30mm}|C{40mm}|C{40mm}| }
    \hline
    Features & Two-step iterative approach \cite{Choi15b} & Our proposed method\\ \hline
    Type of input surfaces & Simply-connected open surfaces & Simply-connected open surfaces\\ \hline
    Initial map & Disk harmonic map & double covering followed by a fast spherical conformal map\\ \hline
    Enforcement of boundary when computing the initial map & Yes & No \\ \hline
    Method for correcting the conformality distortion & Step 1: Use the Cayley transform and work on the upper half plane \newline Step 2: Iterative reflections along the unit circle until convergence &  One-step normalization and composition of quasi-conformal maps \\ \hline
    Output & Unit disk & Unit disk\\ \hline
    Bijectivity & Yes & Yes\\ \hline
    Iterations required? & Yes & No\\ \hline
    \end{tabular}
    \caption{Features of the two-step iterative approach \cite{Choi15b} and our proposed method.}
    \label{difference}
\end{table}

Before explaining our proposed algorithm in details, we point out the major differences between our proposed method and the two-step iterative approach \cite{Choi15b}. Table \ref{difference} highlights the main features of our proposed method and the two-step iterative approach \cite{Choi15b} for disk conformal parameterizations. The two-step iterative approach makes use of the disk harmonic map as an initial map, with the arc-length parameterized circular boundary constraint. This introduces large conformality distortions in the initial map. To correct the conformality distortion, two further steps are required in \cite{Choi15b}. First, the Cayley transform is applied to map the initial disk onto the upper half plane for correcting the distortion at the inner region of the disk. Then, iterative reflections along the unit circle are applied for correcting the distortion near the boundary of the disk until convergence.

In contrast, our proposed fast method primarily consists of only two stages. In the first stage, we find an initial planar parameterization via double covering followed by a spherical conformal map. Since there is no enforcement of the boundary in computing the initial planar parameterization, the conformality distortion of our initial map is much lower than that by the disk harmonic map. In the second stage, we enforce the circular boundary, and then alleviate the conformality distortion as well as achieving the bijectivity using quasi-conformal theories. The absence of iterations in our proposed algorithm attributes to the significant enhancement in the computational time when compared with the two-step iterative approach \cite{Choi15b}.

In the following, we explain the two stages of our proposed algorithm in details.

\subsection{Finding an initial map via double covering}
Instead of directly computing the map $f:M \to \mathbb{D}$ from a simply-connected open surface $M$ to the unit disk $\mathbb{D}$, we tackle the problem by using a simple double covering technique. The double covering technique was also suggested in \cite{Gu03} to compute conformal gradient field of surfaces with boundaries. In the following, we discuss the construction in the continuous setting. Specifically, we construct a genus-0 closed surface $\widetilde{M}$ by the following method. First, we duplicate $M$ and change its orientation. Denote the new copy by $M'$. Then we identify the boundaries of the two surfaces:
\begin{equation}
\partial M \longleftrightarrow \partial M'.
\end{equation}
By the above identification, the two surfaces $M$ and $M'$ are glued along the two boundaries. Note that here we do not identify the interior of the two surfaces $M$ and $M'$. As a result, a closed surface is formed. Denote the new surface by $\widetilde{M}$. It can be easily noted that since $M$ and $M'$ are simply-connected open surfaces, the new surface $\widetilde{M}$ is a genus-0 closed surface. More explicitly, denote by $K$ and $\kappa_g$ the Gaussian curvature and geodesic curvature. Assume that we slightly edit the boundary parts of $M$ and $M'$ so that $\widetilde{M}$ is smooth. Then by the Gauss-Bonnet theorem, we have
\begin{equation}
\int_M K dA + \int_{\partial M} \kappa_g ds = 2\pi \chi(M) =  2\pi
\end{equation}
and
\begin{equation}
\int_{M'} K dA + \int_{\partial M'} \kappa_g ds = 2\pi \chi(M') =  2\pi.
\end{equation}
Hence, we have
\begin{equation}
\begin{split}
2\pi \chi(\widetilde{M}) &= \int_{\widetilde{M}} K dA \\
&= \int_M K dA + \int_{M'} K dA \\
&= \left(2\pi - \int_{\partial M} \kappa_g ds\right) +  \left(2\pi - \int_{\partial M'} \kappa_g ds\right) \\
&= 4\pi - \int_{\partial M} \kappa_g ds + \int_{\partial M} \kappa_g ds \\
&= 4\pi.
\end{split}
\end{equation}
Therefore, the new surface $\widetilde{M}$ has Euler characteristic $\chi(\widetilde{M})=2$, which implies that it is a genus-0 closed surface. As a remark, in the discrete case, the unsmooth part caused by the double covering does not cause any difficulty in our algorithm since we are only considering the angle structure of the glued mesh. The details of the combinatorial argument are explained in Section \ref{implementation}.

After obtaining $\widetilde{M}$ by the abovementioned double covering technique, we look for a conformal map that maps $\widetilde{M}$ to the unit sphere. By the uniformization theorem, every genus-0 closed surface is conformally equivalent to the unit sphere. Hence, the existence of such a conformal map is theoretically guaranteed. In \cite{Choi15}, Choi \emph{et al.} proposed a fast algorithm for computing a conformal map between genus-0 closed surfaces and the unit sphere. The algorithm consists of two steps, and in each step one sparse symmetric positive definite linear system is to be solved. In the following, we briefly describe the mentioned spherical conformal parameterization algorithm.

The {\it harmonic energy functional} of a map $\psi :N \to \mathbb{S}^2$ from a genus-0 closed surface $N$ to the unit sphere is defined as
\begin{equation}
 E(\psi) = \int_N ||\nabla \psi||^2 dv_N.
\end{equation}
In the space of mappings, the critical points of $E(\psi)$ are called {\it harmonic mappings}. For genus-0 closed surfaces, conformal maps are equivalent to harmonic maps \cite{Jost11}. Therefore, to find a spherical conformal map, we can consider solving the following Laplace equation
\begin{equation}\label{eqt:original}
 \Delta^T \psi = 0
\end{equation}
subject to the spherical constraint $||\psi|| = 1$, where $\Delta^T \psi$ is the tangential component of $\Delta \psi$ on the tangent plane of $\mathbb{S}^2$. Note that this problem is nonlinear. In \cite{Choi15}, the authors linearize this problem by solving the equation on the complex plane:
\begin{equation} \label{eqt:laplace}
 \Delta \phi = 0
\end{equation}
given three point boundary correspondences $\phi(a_i) = b_i$, where $a_i,b_i \in \mathbb{C}$ for $i = 1,2,3$. Note that $\Delta^T \phi = \Delta \phi = 0$ since the target domain is now $\mathbb{C}$. Now the problem (\ref{eqt:laplace}) becomes linear since $\Delta \phi$ is linear and the nonlinear constraint $||\psi||=1$ in the original problem (\ref{eqt:original}) is removed. After solving the problem (\ref{eqt:laplace}), the inverse stereographic projection is applied for obtaining a spherical parameterization. Note that in the discrete case, the conformality of the inner region on the complex plane is negligible but that of the outer region on the complex plane is quite large. Correspondingly, the conformality distortion near the North pole of the sphere is quite large. Therefore, to correct the conformality distortion near the North pole, the authors in \cite{Choi15} propose to apply the South-pole stereographic projection to project the sphere onto the complex plane. Unlike the result obtained by solving Equation (\ref{eqt:laplace}), the part with high conformality distortion is now at the inner region on the plane. By fixing the outermost region and composing the map with a suitable quasi-conformal map, the distortion of the inner region can be corrected. Finally, by the inverse South-pole stereographic projection, a bijective spherical conformal parameterization with negligibly low distortions can be obtained. Readers are referred to \cite{Schoen97} and \cite{Choi15} for more details of the harmonic map theory and the abovementioned algorithm respectively.

The combination of the double covering technique and the fast spherical conformal parameterization algorithm in \cite{Choi15} is particularly advantageous. It should be noted that because of the symmetry of the double covered surface, half of the entries in the coefficient matrix of the discretization of the Laplace Equation (\ref{eqt:laplace}) are duplicated. Therefore, even we have doubled the size of the problem under the double covering technique, we can save half of the computational cost of the coefficient matrix by only computing half of the entries. Moreover, the spherical conformal parameterization algorithm in \cite{Choi15} involves solving only two sparse symmetric positive definite systems of equations. Therefore, the computation is still highly efficient.

After finding a spherical conformal map $\tilde{f}: \widetilde{M} \to \mathbb{S}^2$ for the glued surface $\widetilde{M}$ using the parameterization algorithm, note that by symmetry, we can separate the unit sphere into two parts, each of which exactly corresponds to one of $M$ and $M'$. Since our goal is to find a disk conformal map $f:M \to \mathbb{D}$, we put our focus on only one of the two parts. Now, we apply a M\"obius transformation on $\mathbb{S}^2$ so that the two parts become the northern and southern hemispheres of $\mathbb{S}^2$. After that, by applying the stereographic projection $P: \mathbb{S}^2 \to \overline{\mathbb{C}}$ defined by
\begin{equation}
 P(x,y,z) = \frac{x}{1-z} + i \frac{y}{1-z},
\end{equation}
the southern hemisphere is mapped onto the open unit disk $\mathbb{D}$. Since the M\"obius transformation and the stereographic projection are both conformal mappings, the combination of the above steps provides a conformal map $f:M \to \mathbb{D}$.

Theoretically, by the symmetry of the double covered surface, the boundary of the planar region obtained by the above stereographic projection should be a perfect unit circle. However, in the discrete case, due to irregular triangulations of the meshes and the conformality distortions of the map, the boundary is usually different from a perfect circle, as suggested in the experimental results in \cite{Choi15b}. In other words, the planar region $R$ we obtained after applying the stereographic projection may not be a unit disk. An illustration is given in Figure \ref{fig:foot_disk_unnormalized}. To solve this issue, we need one further step to enforce the circular boundary, at the same time maintaining low conformality distortions and preserving the bijectivity of the parameterization.

\begin{figure}[t]
\begin{center}
   \includegraphics[width=0.5\linewidth]{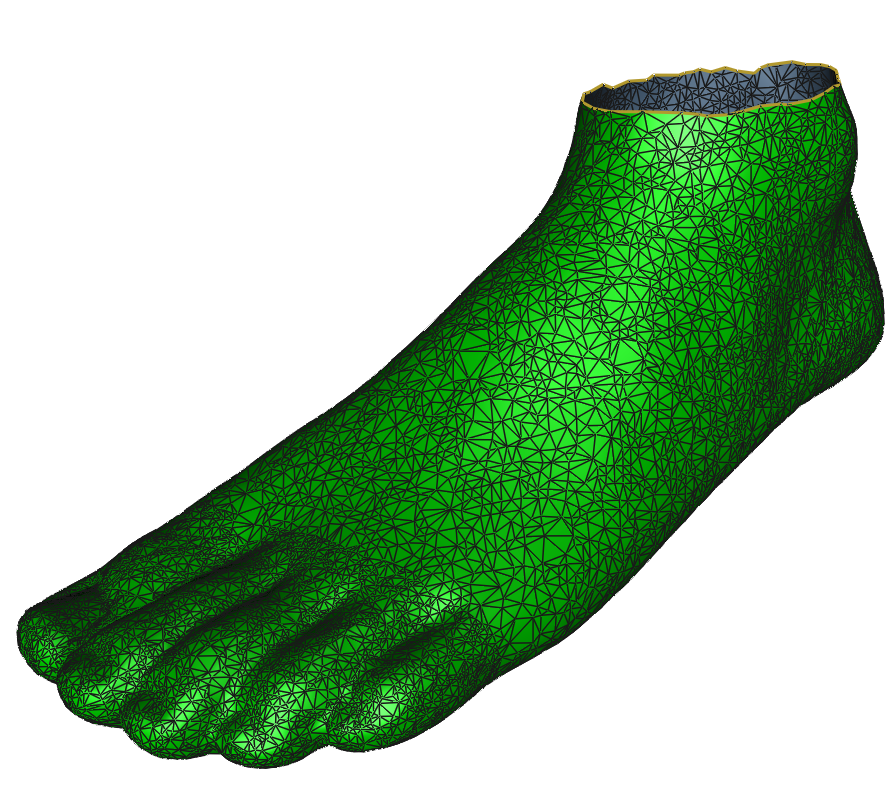}
   \includegraphics[width=0.45\linewidth]{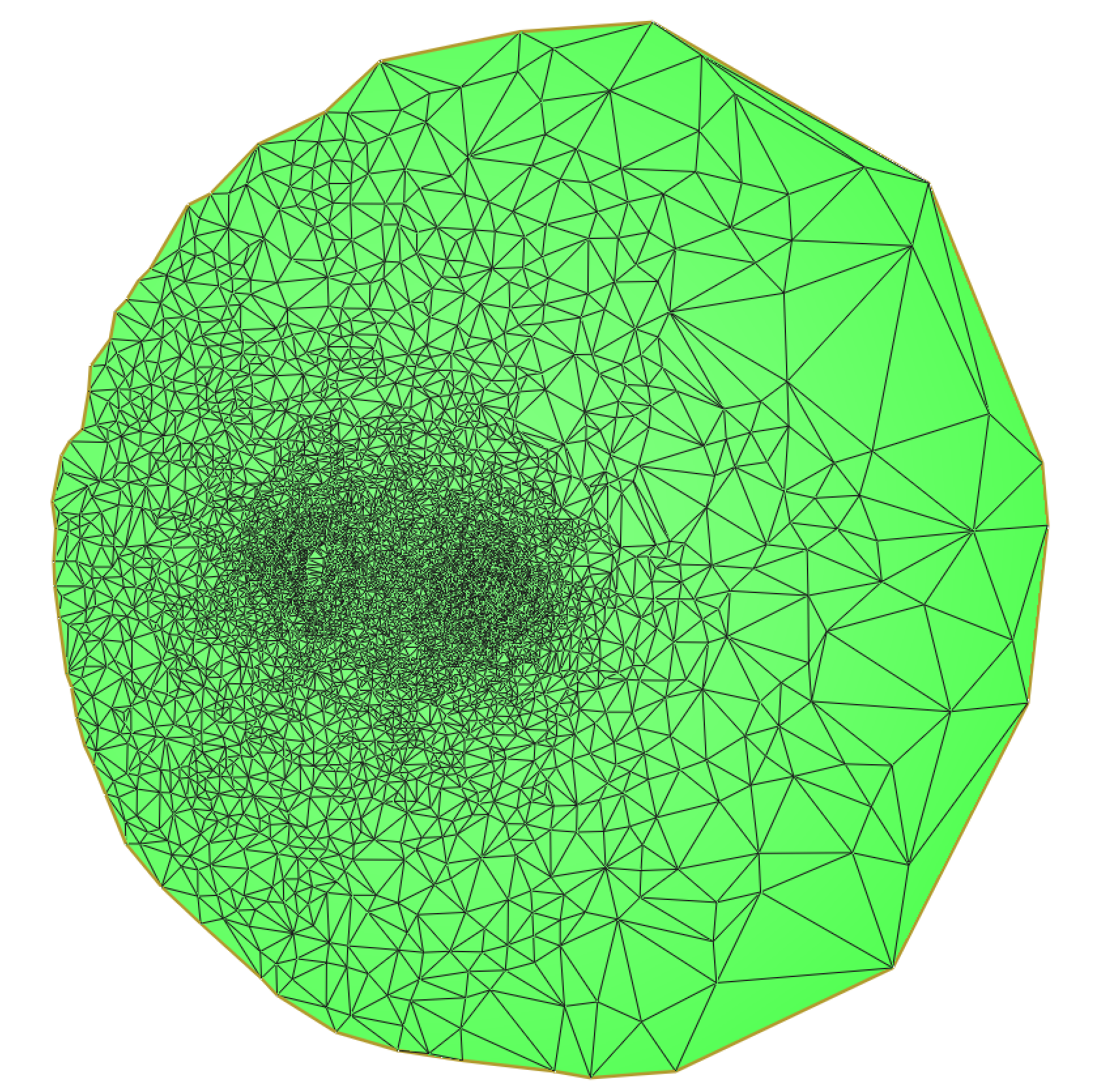}
\end{center}
   \caption{A simply-connected open human foot model and the planar conformal parameterization obtained using double covering followed by spherical conformal map and stereographic projection. It can be observed that the resulting boundary is not perfectly circular.}
\label{fig:foot_disk_unnormalized}
\end{figure}

\subsection{Enforcing the circular boundary to achieve a bijective disk conformal parameterization}
\begin{figure}[t]
\begin{center}
   \includegraphics[width=0.8\linewidth]{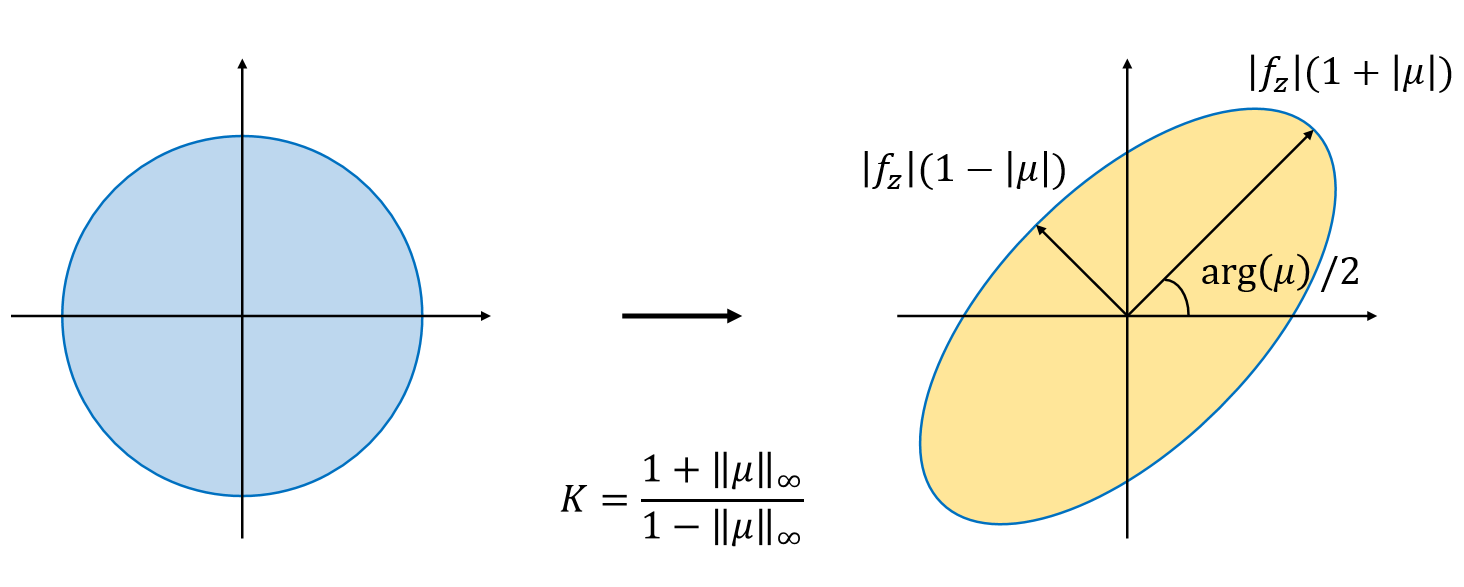}
\end{center}
   \caption{An illustration of quasi-conformal mappings. All information of a quasi-conformal map can be determined by the Beltrami coefficient $\mu$.}
\label{fig:qc}
\end{figure}

To control the conformality distortion and the bijectivity, our idea is to normalize the boundary and then compose the map with a \emph{quasi-conformal map}. Quasi-conformal maps are a generalization of conformal maps, which are orientation preserving homeomorphisms between Riemann surfaces with bounded conformality distortions. Intuitively, a conformal mapping maps infinitesimal circles to infinitesimal circles, while a quasi-conformal mapping maps infinitesimal circles to infinitesimal ellipses with bounded eccentricity \cite{Gardiner00}. Mathematically, a \emph{quasi-conformal map} $f:\mathbb{C} \to \mathbb{C}$ satisfies the Beltrami equation
\begin{equation}\label{beltramieqt}
 \frac{\partial f}{\partial \bar{z}} = \mu(z) \frac{\partial f}{\partial z}
\end{equation}
for some complex-valued functions $\mu$ with $\| \mu \|_\infty<1$. $\mu$ is called the \emph{Beltrami coefficient} of $f$.

The Beltrami coefficient captures the important information of the mapping $f$. For instance, the angles and the magnitudes of both the maximal magnification and the maximal shrinkage can be easily determined by the Beltrami coefficient $\mu$ (see Figure \ref{fig:qc}). Specifically, the angle of the maximal magnification is $arg(\mu(p))/2$ with the magnifying factor $1 + |\mu(p)|$, and the angle of the maximal shrinkage is the orthogonal angle $(arg(\mu(p)) - \pi)/2$ with the shrinking factor $1 - |\mu(p)|$. The maximal dilation of $f$ is given by:
\begin{equation}
K(f) = \frac{1+\|\mu\|_\infty}{1-\|\mu\|_\infty}.
\end{equation}
It is also noteworthy that $f$ is conformal around a small neighborhood of $p$ if and only if $\mu(p) = 0$. Hence, $|\mu|$ is a good indicator of the angular distortions of a mapping.

In fact, the norm of the Beltrami coefficient $\mu$ is not only related to the conformality distortion but also the bijectivity of the associated quasi-conformal mapping, as explained by the following theorem:
\begin{thm}\label{bijectivity}
If $f:M \to \mathbb{D}$ is a $C^1$ map satisfying $\|\mu_f\|_{\infty} <1$, then $f$ is bijective.
\end{thm}
\begin{proof}
 See \cite{Gardiner00}. \qed
\end{proof}

This theorem can be explained with the aid of the Jacobian of $f$. The Jacobian $J_f$ of $f$ is given by
\begin{equation}
J_f = \left| \frac{\partial f}{\partial z} \right|^2 (1-|\mu_f|^2).
\end{equation}
Suppose $\|\mu_f\|_{\infty} <1$, then we have $\left| \frac{\partial f}{\partial z} \right|^2 \neq 0$ and $(1-|\mu_f|^2)>0$. Therefore, $J(f)$ is positive everywhere. Since $\mathbb{D}$ is simply-connected and $f$ is proper, we can conclude that $f$ is a diffeomorphism. In fact, $f$ is a universal covering map of degree 1. Therefore, $f$ must be bijective. One important consequence of Theorem \ref{bijectivity} is that we can easily achieve the bijectivity of a quasi-conformal map by enforcing its associated Beltrami coefficient to be with supremum norm less than 1.

Moreover, it is possible for us to reconstruct a mapping by a given Beltrami coefficient, as explained by the following theorem:
\begin{thm}\label{correspondence}
Let $M_1$ and $M_2$ be two simply-connected open surfaces. Given 2-point correspondences, every Beltrami coefficient $\mu$ with $\|\mu\| <1$ is associated with a unique quasi-conformal homeomorphism $f: M_1 \to M_2$.
\end{thm}
\begin{proof}
 See \cite{Gardiner00}. \qed
\end{proof}

For the aspect of numerical computations, Lui \emph{et al.} \cite{Lui13} proposed the linear Beltrami solver ({\bf LBS}), a fast algorithm for reconstructing a quasi-conformal map on a rectangular domain from a given Beltrami coefficient. The key idea of {\bf LBS} is as follows.

By expanding the Beltrami Equation (\ref{beltramieqt}), we have
\begin{equation}
\mu_g = \frac{(u_x - v_y) + i(v_x + u_y)}{(u_x + v_y) + i(v_x - u_y)}.
\end{equation}
Suppose $\mu(f) = \rho + i \tau$. Then, $v_x$ and $v_y$ can be expressed as linear combinations of $u_x$ and $u_y$:
\begin{equation}
 \begin{array}{cc}
  -v_y &= \alpha_1 u_x + \alpha_2 u_y,\\
  v_x &= \alpha_2 u_x + \alpha_3 u_y,\\
  \end{array}
\end{equation}
where
\begin{equation}
\label{eqt:alpha123}
\begin{split}
  \alpha_1 &= \frac{(\rho -1)^2 + \eta^2}{1-\rho^2 - \eta^2},\\
  \alpha_2 &= -\frac{2\eta}{1-\rho^2 - \eta^2}, \\
  \alpha_3 &= \frac{1+2\rho+\rho^2 +\eta^2}{1-\rho^2 - \eta^2}. \\
\end{split}
\end{equation}
Similarly, we can express $u_x$ and $u_y$ as linear combinations of $v_x$ and $v_y$:
\begin{equation}
 \begin{array}{cc}
  -u_y &= \alpha_1 v_x + \alpha_2 v_y,\\
  u_x &= \alpha_2 v_x + \alpha_3 v_y.\\
  \end{array}
\end{equation}
Hence, to solve for a quasi-conformal map, it remains to solve
\begin{equation}\label{eqt:BeltramiPDE}
\nabla \cdot \left(A \left(\begin{array}{c}
u_x\\
u_y \end{array}\right) \right) = 0\ \ \mathrm{and}\ \ \nabla \cdot \left(A \left(\begin{array}{c}
v_x\\
v_y \end{array}\right) \right) = 0
\end{equation}
where
$A = \left( \begin{array}{cc} \alpha_1 & \alpha_2\\
\alpha_2 & \alpha_3 \end{array}\right).$

In the discrete case, the above elliptic PDEs (\ref{eqt:BeltramiPDE}) can be discretized into sparse linear systems. For details, please refer to \cite{Lui13}. In the following discussion, we denote the quasi-conformal map associated with the Beltrami coefficient $\mu$ obtained by {\bf LBS} by $\mathbf{LBS}(\mu)$.

Another important property of quasi-conformal mappings is about their composition mappings. In fact, the Beltrami coefficient of a composition mapping can be explicitly expressed in terms of the Beltrami coefficients of the original mappings.
\begin{thm}\label{composition}
Let $f: \Omega \subset \mathbb{C} \to f(\Omega)$ and $g: f(\Omega) \to \mathbb{C}$ be two quasi-conformal mappings. Then the Beltrami coefficient of $g \circ f$ is given by
\begin{equation}
\mu_{g \circ f} = \frac{\mu_f+(\overline{f_z}/f_z) (\mu_g \circ f)}{1+(\overline{f_z}/f_z)  \overline{\mu_f} (\mu_g \circ f)}.
\end{equation}
In particular, if $\mu_{f^{-1}} \equiv  \mu_g$, then since $\mu_{f^{-1}} \circ f = -(f_z / \overline{f_z}) \mu_f$, we have
\begin{equation}
\mu_{g \circ f} \equiv \frac{\mu_f+(\overline{f_z}/f_z) ((-f_z/\overline{f_z}) \mu_f)}{1+(\overline{f_z}/f_z)  \overline{\mu_f} ((-f_z/\overline{f_z}) \mu_f)} \equiv 0.
\end{equation}
Hence $g \circ f$ is conformal.
\end{thm}
\begin{proof}
 See \cite{Gardiner00}. \qed
\end{proof}

In other words, by composing two quasi-conformal maps whose Beltrami coefficients satisfy the above condition, one can immediately obtain a conformal map. This observation motivates the following step.

To enforce the circular boundary of the parameterization, we first normalize the boundary of the region $R$ to the unit circle:
\begin{equation} \label{eqt:normalization}
 v \mapsto \frac{v}{|v|}
\end{equation}
for all $v \in \partial R$. Denote the normalized region by $\widetilde{R}$. Since the vertices near the boundary of the region $R$ may be very dense, a direct normalization of the boundary may cause overlaps of the triangulations as well as geometric distortions on the unit disk. To eliminate the overlaps and the distortions of $\widetilde{R}$, we apply the linear Beltrami solver to construct another quasi-conformal map with the normalized boundary constraints. Then by the composition property, the composition map becomes a conformal map. More specifically, denote the Beltrami coefficient of the mapping $g: \widetilde{R} \to M$ from the normalized planar region to the original surface by $\mu$. We reconstruct a quasi-conformal map with Beltrami coefficient $\mu$ on the unit disk by extending the linear Beltrami solver, so that it is applicable not only on rectangular domains but also circular domains. We compute a map $h: \widetilde{R} \to \mathbb{D}$ by applying the linear Beltrami solver:
\begin{equation}
 h = \textbf{LBS}(\mu)
\end{equation}
with the circular boundary constraint $h(v) = v \text{ for all } v \in \partial \widetilde{R}$.

Note that by the composition property stated in Theorem \ref{composition}, $h \circ g^{-1}$ is a conformal map from the original surface $M$ to the unit disk $\mathbb{D}$. Finally, the bijectivity of the composition map is supported by Theorem \ref{bijectivity}, since the Beltrami coefficient of the composition map is with supremum norm less than 1. This completes the task of finding a bijective disk conformal parameterization. The numerical implementation of our proposed method is explained in Section \ref{implementation}.

\section{Numerical Implementation} \label{implementation}
In this section, we describe the numerical implementation of our proposed algorithm in details. In the discrete case, 3D surfaces are commonly represented by triangular meshes. Discrete analogs of the theories on the smooth surfaces are developed on the triangulations.

We first briefly describe the discrete version of the mentioned double covering technique for obtaining a genus-0 closed mesh. This discretization was also applied in \cite{Gu03} to compute conformal gradient fields of surfaces with boundaries. A triangulation $K=(V,E,F)$ of a smooth simply-connected open surface consists of the vertex set $V$, the edge set $E$ and the triangular face set $F$. Each face can be represented as an ordered triple $[u,v,w]$ where $u,v,w$ are three vertices. Suppose the boundary vertices of $K$ are $\{w_i\}_{i=1}^r$. Denote a duplicated triangulation of $K$ by $K' = (V',E',F')$, with boundary vertices $\{w'_i\}_{i=1}^r$. Let $v,e,f$ denote the number of vertices, edges and faces of $M$ respectively. We duplicate $M$ and denote the new copy by $M'$, with $v',e',f'$ the number of vertices, edges and faces. By Euler's formula, we have
\begin{equation}
 v-e+f = v'-e'+f' = 1.
\end{equation}
Then, we construct a new surface $\widetilde{M}$ by reversing the face orientation of $M'$ and gluing the two boundaries $\partial M$ and $\partial M'$. To reverse the orientation of $K'$, we rearrange the order of the vertices of each face in $F'$ from $[u,v,w]$ to $[u,w,v]$. Then, to glue the two surfaces, we identify the boundary vertices of the two meshes:
\begin{equation}
w_i \longleftrightarrow w'_i \text{ for all } i=1,2,\cdots,r.
\end{equation}
Now, denote the number of vertices, edges and faces of $\widetilde{M}$ by $\tilde{v}$, $\tilde{e}$ and $\tilde{f}$ respectively. It follows that
\begin{equation}
 \begin{array}{ll}
  &\tilde{v} - \tilde{e} + \tilde{f} \\
  =&(v+v' - r) - (e+e' - r) + (f+f') \\
  =& v + v' - e - e' + f + f' \\
  =&2,
 \end{array}
\end{equation}
Hence, the new surface $\widetilde{M}$ is a genus-0 closed surface. The complete double covering procedure is described in Algorithm \ref{alg:double}.

\begin{algorithm}[h]
\KwIn{A triangulation $K=(V,E,F)$ of a simply-connected open surface.}
\KwOut{A genus-0 closed mesh $\widetilde{K} = (\widetilde{V},\widetilde{E},\widetilde{F})$.}
\BlankLine
Duplicate $K$ and denote the copy by $K' = (V',E',F')$\;
Change the order of the vertices of each face in $F'$ from  $[u,v,w]$ to $[u,w,v]$\;
Replace the boundary vertices $w'_i$ by $w_i$ in $E'$ and $F'$ for $i=1,2,\cdots,r$\;
Set $\widetilde{V} = V \cup V' \setminus \{w'_i\}_{i=1}^r$\;
Set $\widetilde{E} = E \cup E'$\;
Set $\widetilde{F} = E \cup E'$\;
\caption{The double covering technique in our initial step.}
\label{alg:double}
\end{algorithm}

Then, we introduce the discretization about harmonic mappings used in the fast spherical conformal parameterization algorithm in \cite{Choi15}. Let $\widetilde{K}$ be the triangulation of a genus-0 closed surface $N$. Let us denote the edge spanned by two vertices $u,v$ on $\widetilde{K}$ by $[u,v]$. The {\it discrete harmonic energy} of $\psi:\widetilde{K} \to \mathbb{S}^2$ is given by
\begin{equation}
E(\psi) = \frac{1}{2} \sum_{[u,v] \in \widetilde{K}} k_{uv} ||\psi(u)-\psi(v)||^2,
\end{equation}
where $k_{uv} = \cot \alpha + \cot \beta$ with $\alpha,\beta$ being the angles opposite to the edge $[u,v]$. This is known as the cotangent formula \cite{Pinkall93}.
With the discrete harmonic energy, we can appropriately discretize the Laplace-Beltrami operator as
\begin{equation}
\Delta \psi(v_i) = \sum_{v_j \in N(v_i)} k_{v_i v_j} (\psi(v_j)-\psi(v_i)),
\end{equation}
where $N(v_i)$ denotes the set of the 1-ring neighboring vertices of $v_i$. Therefore, the Laplace Equation (\ref{eqt:laplace}) becomes a linear system in the form
\begin{equation}
\label{eqt:laplace_linear}
Az = b,
\end{equation}
where $A$ is a square matrix satisfying
\begin{equation}
\displaystyle
  A_{ij} = \left\{
 \begin{array}{ll}
 k_{v_i v_j} & \text{ if } i \neq j, \\
 -\sum_{v_t \in N(v_i)} k_{v_i v_t} & \text{ if } i = j\\
 \end{array} \right. \text{ for all }i,j.
\end{equation}
For the choice of the points $a_1,a_2,a_3,b_1,b_2,b_3$ in the boundary constraint of Equation (\ref{eqt:laplace}), we choose a triangular face $[a_1,a_2,a_3]$ on the triangulation $K$ and a triangle $[b_1,b_2,b_3]$ on the complex plane that shares the same angle structure as $[a_1,a_2,a_3]$. Note that the above linear system is sparse and symmetric positive definite. Therefore, it can be efficiently solved.

It is noteworthy that due to the symmetry of the double covered surface, for every edge $[u,v]$ in the triangulation $K$, there exists a unique edge $[u',v']$ in the duplicated triangulation $K'$ such that
\begin{equation}
 k_{uv} = \cot \alpha + \cot \beta = \cot \alpha' + \cot \beta' = k_{u'v'},
\end{equation}
where $\alpha,\beta$ are the angles opposite to the edge $[u,v]$ and $\alpha',\beta'$ are the angles opposite to the edge $[u',v']$. Therefore, only half of the vertices and faces are needed for computing the whole coefficient matrix $A$ to solve the Laplace Equation (\ref{eqt:laplace}). More explicitly, Equation (\ref{eqt:laplace_linear}) can be expressed as the following form:
\begin{equation}
 \begin{pmatrix}
  B & * &  0 \\
  * & * & * \\
  0 & * & B
 \end{pmatrix}
 \begin{pmatrix}
  z_1 \\
  \vdots \\
  z_{|V|-r} \\
  \zeta_1 \\
  \vdots \\
  \zeta_r \\
  z'_1 \\
  \vdots \\
  z'_{|V|-r} \\
 \end{pmatrix}
= b.
\end{equation}
Here, $\{z_i\}_{i=1}^{|V|-r}$ and $\{z'_i\}_{i=1}^{|V|-r}$ are respectively the coordinates of the non-boundary vertices of $K$ and $K'$, $\{\zeta_i\}_{i=1}^{r}$ are the coordinates of the glued vertices, and $B$ is a $(|V|-r) \times (|V|-r)$ sparse symmetric positive definite matrix. It follows that we can save half of the computational cost in finding all the cotangent weights $k_{uv}$ in $A$. Hence, the computation of the spherical conformal map is efficient even the number of vertices and faces is doubled under the double covering step.

Another important mathematical tool in our proposed algorithm is the quasi-conformal mapping. Quasi-conformal mappings are closely related to the Beltrami coefficients. It is important to establish algorithms for computing the Beltrami coefficient associated with a given quasi-conformal map, as well as for computing the quasi-conformal map associated with a given Beltrami coefficient.

We first focus on the computation of the Beltrami coefficients. In the discrete case, suppose $K_1, K_2 \subset \mathbb{R}^3$ are two triangular meshes with the same number of vertices, faces and edges, and $f:K_1 \to K_2$ is an orientation preserving piecewise linear homeomorphism. It is common to discretize the Beltrami coefficient on the triangular faces. To compute the Beltrami coefficient $\mu_f$ associated with $f$, we compute the partial derivatives on every face $T_1$ on $K_1$.

Suppose $T_1$ on $K_1$ corresponds to a triangular face $T_2$ on $K_2$ under the mapping $f$. The approximation of $\mu_f$ on $T_1$ can be computed using the coordinates of the six vertices of $T_1$ and $T_2$. Since the triangulations are piecewise linear, we can place $T_1$ and $T_2$ on $\mathbb{R}^2$ using suitable rotations and translations to simplify the computations. Hence, without loss of generality, we can assume that $T_1$ and $T_2$ are on $\mathbb{R}^2$. Specifically, suppose $T_1 = [a_1+i\ b_1, a_2+ i\ b_2, a_3 + i\ b_3]$ and $T_2 = [w_1, w_2, w_3]$, where $a_1,a_2,a_3,b_1,b_2,b_3 \in \mathbb{R}$, and $w_1,w_2,w_3 \in \mathbb{C}$. Recall that $\frac{\partial f}{\partial \bar{z}} = \frac{1}{2} \left(\frac{\partial f}{\partial x} + i \frac{\partial f}{\partial y}\right)$ and $\frac{\partial f}{\partial z} = \frac{1}{2} \left(\frac{\partial f}{\partial x} - i \frac{\partial f}{\partial y}\right)$. Hence, to discretize the Beltrami coefficient, we only need to compute $\frac{\partial}{\partial x}$ and $\frac{\partial}{\partial y}$ on every triangular face $T_1$. It is natural to use the differences between the vertex coordinates for the approximation. We define
\begin{equation}
D_x = \frac{1}{2 Area(T_1)} \left( \begin{array}{c} b_3-b_2 \\ b_1-b_3 \\ b_2-b_1 \end{array} \right)^t \ \ \mathrm{and}\ \
D_y = -\frac{1}{2 Area(T_1)} \left( \begin{array}{c} a_3-a_2 \\ a_1-a_3 \\ a_2-a_1 \end{array} \right)^t.
\end{equation}
Then, we can compute the Beltrami coefficient
\begin{equation}
\mu_f(z) = \frac{\partial f}{\partial \bar{z}} \left/ \frac{\partial f}{\partial z}\right.
\end{equation}
on $T_1$ by
\begin{equation}\label{eqt:discreteBC}
\mu_f(T_1) = \frac{\frac{1}{2} \left(D_x + i\ D_y \right) \left( \begin{array}{c} w_1 \\ w_2 \\ w_3 \end{array} \right) }{\frac{1}{2} \left(D_x - i\ D_y \right) \left( \begin{array}{c} w_1 \\ w_2 \\ w_3 \end{array} \right) }.
\end{equation}

This approximation is easy to compute. Hence, it is convenient to obtain the Beltrami coefficient associated with a given quasi-conformal map in the discrete case. A relatively complicated task is to compute the quasi-conformal map $f$ associated with a given Beltrami coefficient $\mu_f$. To achieve this, we apply the {\bf LBS} \cite{Lui13} to reconstruct a quasi-conformal map from a given Beltrami coefficient, with the boundary vertices of the disk fixed. We now briefly explain the discretization of the {\bf LBS}.

Recall that the quasi-conformal map associated with a given Beltrami coefficient can be obtained by solving Equation (\ref{eqt:BeltramiPDE}). The key idea of {\bf LBS} is to discretize Equation (\ref{eqt:BeltramiPDE}) into sparse SPD linear systems of equations so that the solution can be efficiently computed.

For each vertex $v_i$, let $N_i$ be the collection of the neighboring faces attached to $v_i$. Let $T = [v_i,v_j, v_k]$ be a face ,$w_i = f(v_i), w_j = f(v_j)$ and $w_k = f(v_k)$. Suppose $v_l = g_l + i\ h_l$ and $w_l = s_l + i\ t_l$, for $l=i,j,k$. We further denote the Beltrami coefficient of the face $T$ obtained from Equation (\ref{eqt:discreteBC}) by $\mu_f(T) = \rho_T + i\ \eta_T$. Equation (\ref{eqt:alpha123}) can be discretized as follows:
\begin{equation}
\begin{split}
 \alpha_1(T) &= \frac{(\rho_T -1)^2 + \eta_T^2}{1-\rho_T^2 - \eta_T^2},\\
 \alpha_2(T) &= -\frac{2\eta_T}{1-\rho_T^2 - \eta_T^2},\\
 \alpha_3(T) &= \frac{1+2\rho_T+\rho_T^2 +\eta_T^2}{1-\rho_T^2 - \eta_T^2}.\\
\end{split}
\end{equation}

Then, Lui {\it et al.} \cite{Lui13} introduced the discrete divergence to compute the divergence operator. Define $A_i^T, A_j^T, A_k^T, B_i^T, B_j^T, B_k^T$ by
\begin{equation}
 \left(\begin{array}{c}
  A_i^T\\
  A_j^T\\
  A_k^T
 \end{array}\right) = \frac{1}{Area(T)}
\left(\begin{array}{c}
  h_j-h_k\\
  h_k-h_i\\
  h_i-h_j
 \end{array}\right) \text{ and }
 \left(\begin{array}{c}
  B_i^T\\
  B_j^T\\
  B_k^T
 \end{array}\right) = \frac{1}{Area(T)}
\left(\begin{array}{c}
  g_k-g_j\\
  g_i-g_k\\
  g_j-g_i
 \end{array}\right).
\end{equation}
Then the discrete divergence of a discrete vector field $\vec{V} = (V_1,V_2)$ on the triangular faces can be defined by
\begin{equation}
  Div(\vec{V})(v_i) = \sum_{T \in N_i} A_i^T V_1(T)+ B_i^T V_2(T).
\end{equation}
With the discrete divergence, Equation (\ref{eqt:BeltramiPDE}) can be discretized into the following linear system:
\begin{equation}\label{eqt:linearB12}
\begin{split}
\sum_{T\in N_i} A_i^T[\alpha_1(T) a_T + \alpha_2(T) b_T]+B_i^T[\alpha_2(T) a_T + \alpha_3(T) b_T]  = 0\\
\sum_{T\in N_i} A_i^T[\alpha_1(T) c_T + \alpha_2(T) d_T]+B_i^T[\alpha_2(T) c_T + \alpha_3(T) d_T] = 0
\end{split}
\end{equation}
for all vertices $v_i$. Here $a_T$, $b_T$, $c_T$ and $d_T$ are certain linear combinations of the $x$-coordinates and $y$-coordinates of the desired quasi-conformal map $f$. Hence, we can obtain the $x$-coordinate and $y$-coordinate function of $f$ by solving the linear system in Equation (\ref{eqt:linearB12}). For more details, please refer to \cite{Lui13}.

With the aforementioned discrete analogs and algorithms, we are ready to describe the numerical implementation of our proposed algorithm for disk conformal parameterizations of simply-connected open surfaces.

Given a triangular mesh $M$ of disk topology, we first apply the double covering technique as described in Algorithm \ref{alg:double}. Then, we apply the fast spherical conformal parameterization algorithm \cite{Choi15} on the glued genus-0 closed mesh $\widetilde{M}$. To obtain a planar parameterization from the obtained parameter sphere $S$, we apply a M\"obius transformation to locate the part corresponding to one of the two copies on the southern hemisphere. After that, by the stereographic projection, the Southern hemisphere is mapped to a planar region $R$ in the complex plane. Note that this planar region is usually a bit different from a perfect disk due to the discretization and approximation errors. Therefore, we normalize the boundary of $R$ to enforce a circular boundary as in Equation (\ref{eqt:normalization}). However, this step may causes overlaps as well as conformality distortions on the normalized region $\widetilde{R}$.

To overcome these issues, we apply the idea of the composition of quasi-conformal maps in Theorem \ref{composition}. We first compute the Beltrami coefficient of the quasi-conformal map  $g:\widetilde{R} \to M$ by solving Equation (\ref{eqt:discreteBC}) on all triangular faces of $\widetilde{R}$. Denote the Beltrami coefficient by $\mu$. We then reconstruct another quasi-conformal map $h: \widetilde{R} \to \mathbb{D}$ with the given Beltrami coefficient $\mu$, using {\bf LBS} \cite{Lui13}. Note that boundary constraints are needed in solving Equation (\ref{eqt:linearB12}). In our case, we give the Dirichlet condition on the whole boundary of the normalized disk $\widetilde{R}$. More explicitly, the boundary condition used in solving Equation (\ref{eqt:linearB12}) for the map $h: \widetilde{R} \to \mathbb{D}$ is
\begin{equation}
 h(v) = v
\end{equation}
for all boundary vertices $v \in \partial \widetilde{R}$. After obtaining $h$, we can conclude that the composition map $f = h \circ g^{-1}$ is the desired disk conformal parameterization. The complete implementation of our algorithm is described in Algorithm \ref{algorithm}.

\begin{algorithm}[h]
\KwIn{A simply-connected open mesh $M$.}
\KwOut{A bijective disk conformal parameterization $f:M \to \mathbb{D}$.}
\BlankLine
Double cover $M$ and obtain a genus-0 closed mesh $\widetilde{M}$\;
Apply the fast spherical conformal parameterization \cite{Choi15} on $\widetilde{M}$ and obtain the parameter sphere $S$\;
Apply a M\"obius transformation on $S$ so that the original surface $M$ corresponds to the southern hemisphere of $S$\;
Using the stereographic projection for the southern hemisphere of $S$, obtain a planar region $R$\;
Normalize the boundary of $R$:
\begin{equation}
    v \mapsto \frac{v}{|v|}
\end{equation}
for all $v \in \partial R$ and denote the normalized region by $\widetilde{R}$\;
Compute the Beltrami coefficient of the map $g:\widetilde{R} \to M$ and denote it by $\mu$\;
Apply the linear Beltrami solver \cite{Lui13} to obtain a map $h: \widetilde{R} \to \mathbb{D}$
\begin{equation}
 h = \textbf{LBS}(\mu)
\end{equation}
with the boundary constraint $h(v) =v$ for all $v \in \partial \widetilde{R}$. The composition map $f = h \circ g^{-1}$ is the desired disk conformal parameterization\;
\caption{Our proposed linear formulation for disk conformal parameterization}
\label{algorithm}
\end{algorithm}

\section{Experimental results} \label{experiments}

\begin{figure}[t]
\begin{center}
   \includegraphics[width=0.45\linewidth]{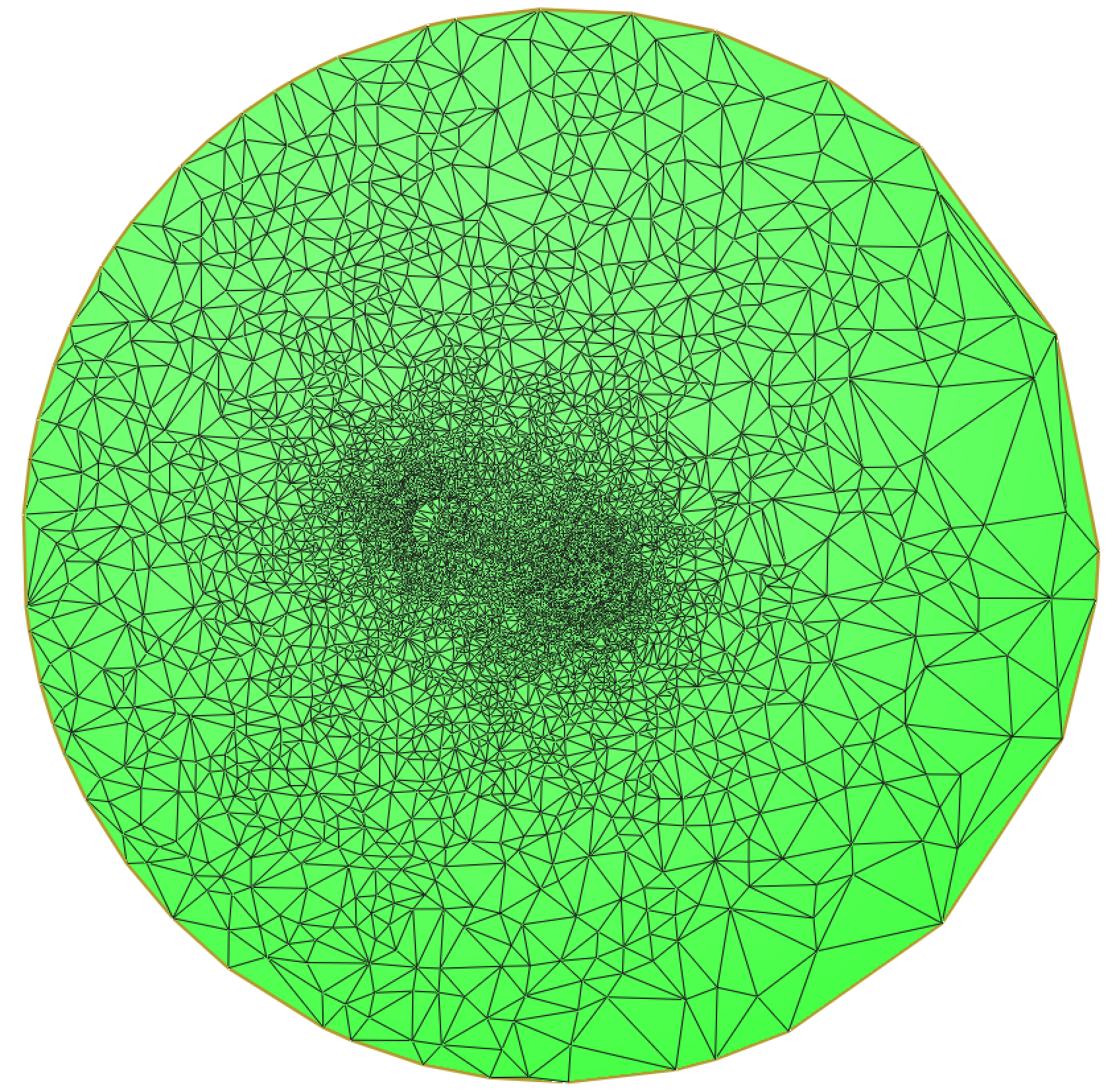}
\end{center}
   \caption{The disk conformal parameterization of the foot model in Figure \ref{fig:foot_disk_unnormalized} using our proposed algorithm.}
\label{fig:foot}
\end{figure}

\begin{figure}[t]
\begin{center}
   \includegraphics[width=0.45\linewidth]{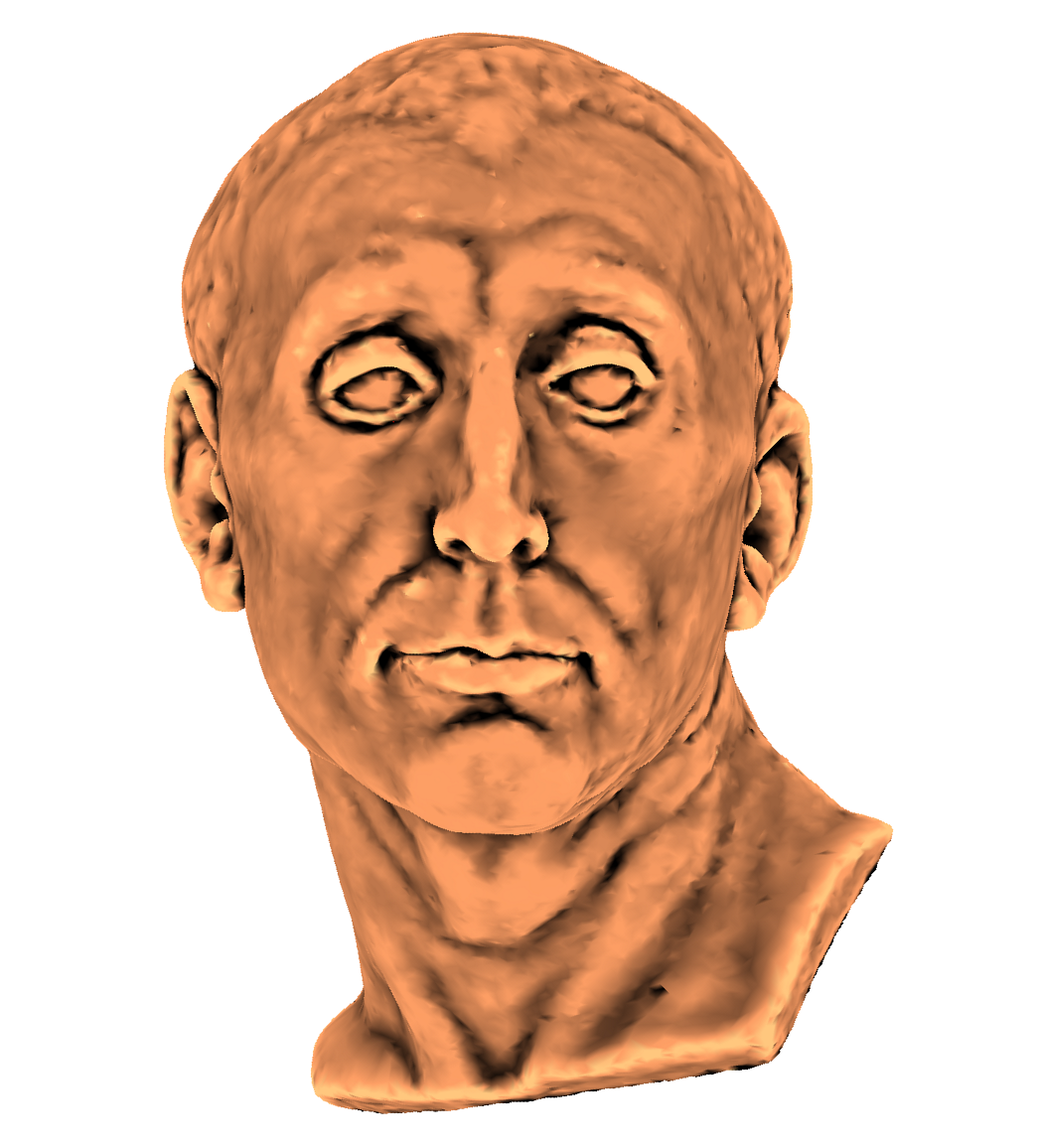}
   \includegraphics[width=0.45\linewidth]{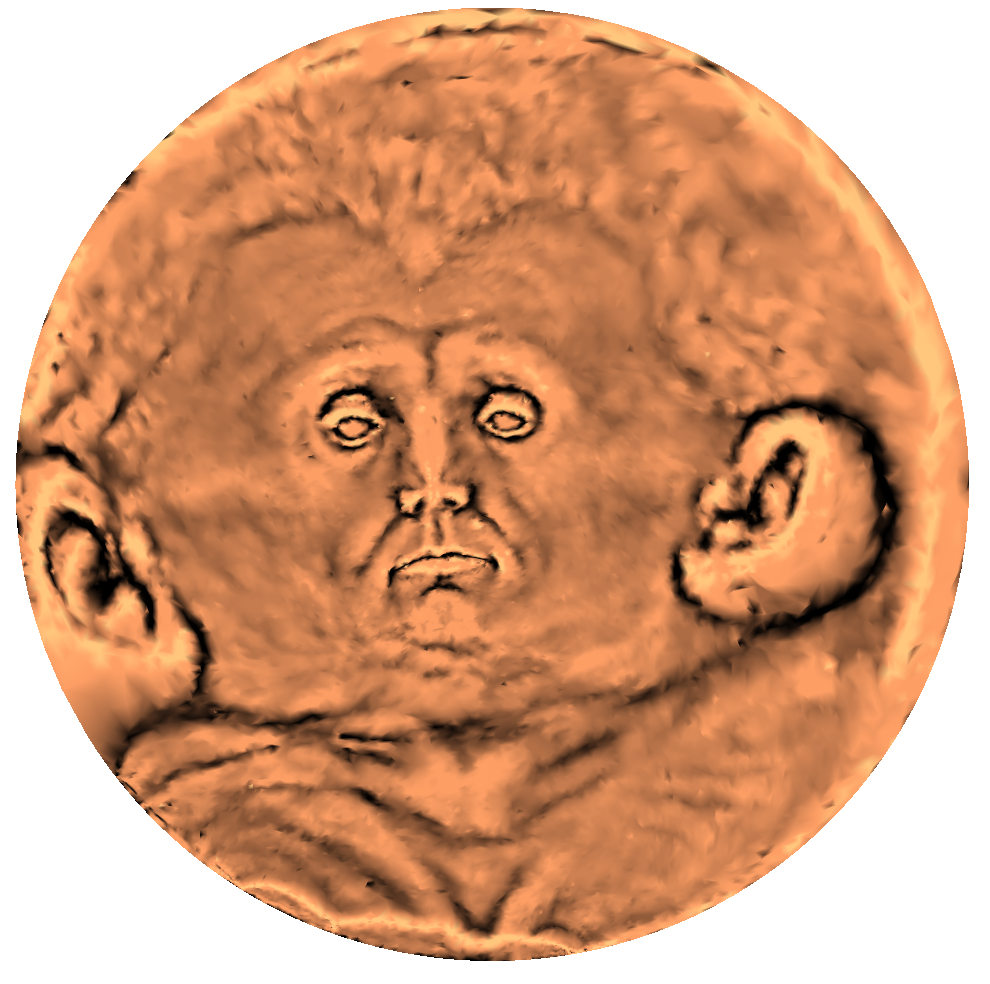}
\end{center}
   \caption{A simply-connected open statue model and the disk conformal parameterization obtained by our proposed algorithm. The color represents the mean curvature of the model.}
\label{fig:statue}
\end{figure}

\begin{figure}[t]
\begin{center}
   \includegraphics[width=0.45\linewidth]{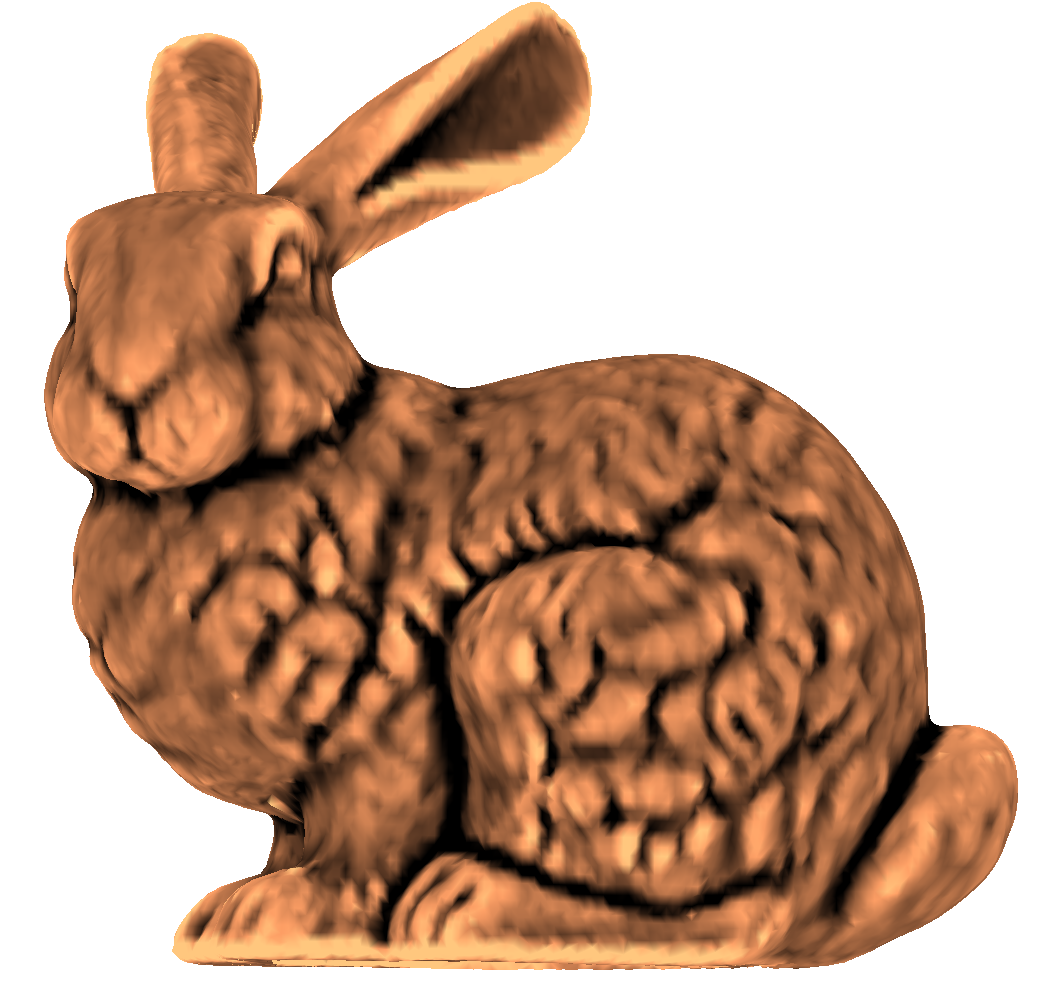}
   \includegraphics[width=0.45\linewidth]{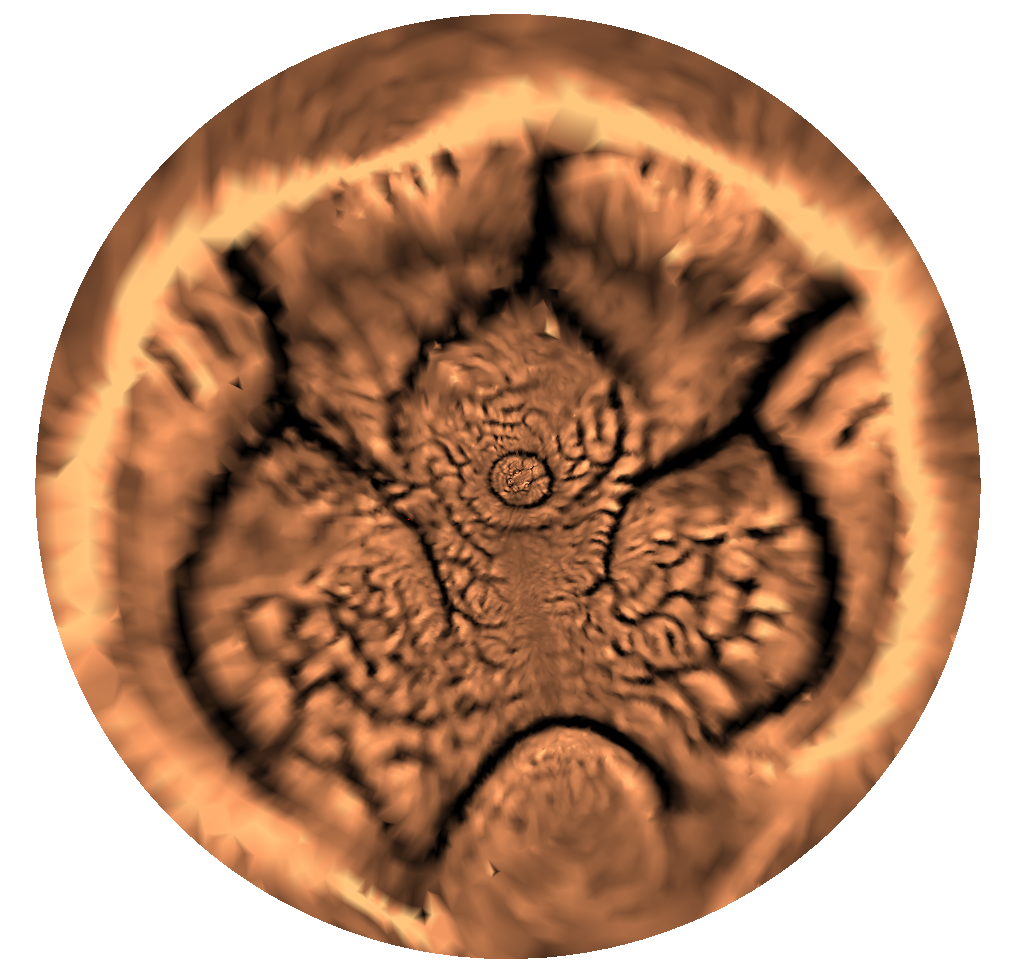}
\end{center}
   \caption{A simply-connected open Stanford bunny model and the disk conformal parameterization obtained by our proposed algorithm. The color represents the mean curvature of the model.}
\label{fig:bunny}
\end{figure}

\begin{figure}[t]
\begin{center}
\includegraphics[width=0.45\linewidth]{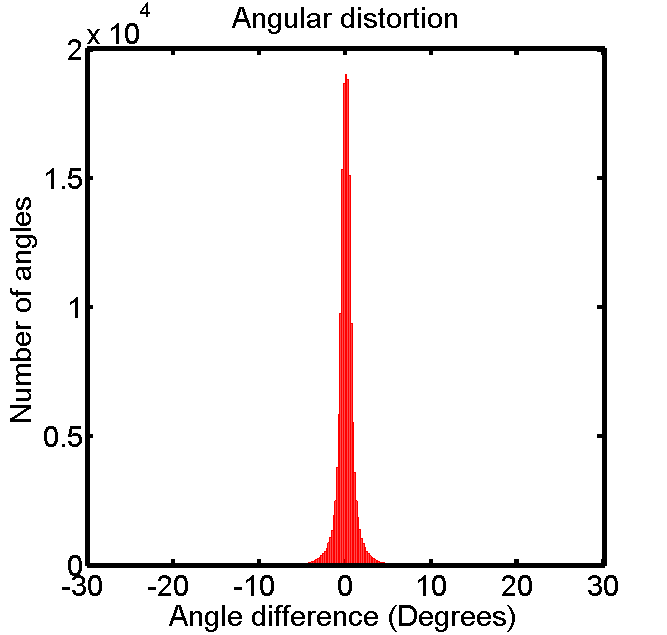}
\includegraphics[width=0.45\linewidth]{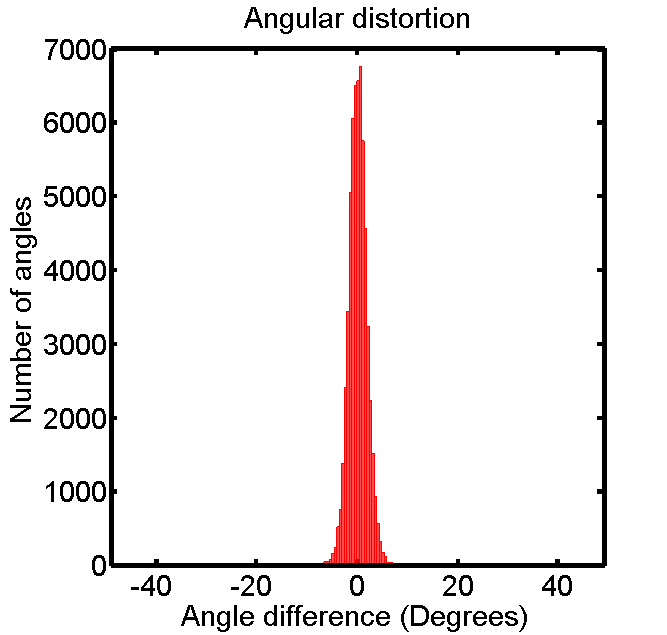}
\includegraphics[width=0.45\linewidth]{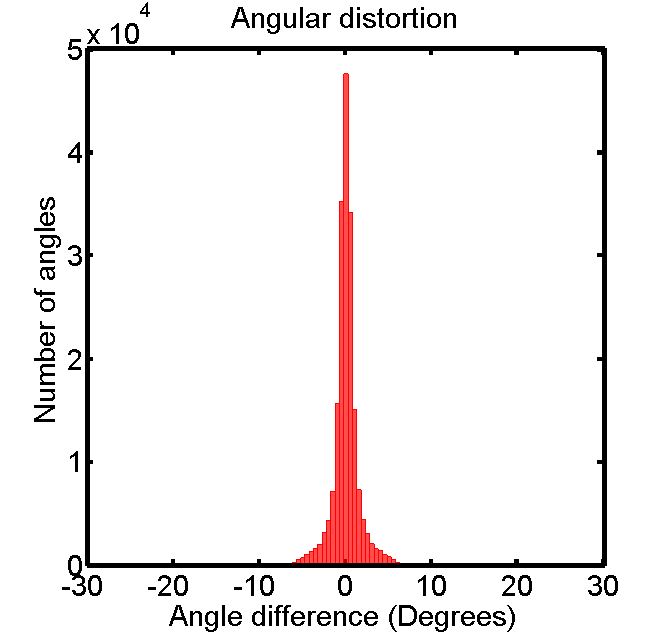}
\end{center}
   \caption{The histograms of the angular distortion of our parameterization algorithm. Here the angular distortion refers to the difference (in degrees) between the new angle after the parameterization and its corresponding original angle on the mesh. Top left: The result of statue. Top right: The result of foot. Bottom: The result of Stanford bunny.}
\label{fig:histogram}
\end{figure}

In this section, we demonstrate the effectiveness of our proposed algorithm using a number of 3D simply-connected open meshes. The meshes are freely adapted from mesh repositories such as the AIM@SHAPE Shape Repository \cite{aim@shape}, the Stanford 3D Scanning Repository \cite{stanford} and the Benchmark for 3D Mesh Segmentation \cite{princeton}. Our proposed algorithm and the two-step iterative approach \cite{Choi15b} are implemented in MATLAB. The sparse linear systems in our proposed algorithm are solved using the backslash operator (\textbackslash) in MATLAB. For the other existing algorithms in our comparison, we adopt available online software for the computation. All experiments are performed on a PC with a 3.40 GHz quad core CPU and 16 GB RAM.

We apply our proposed algorithm on various kinds of simply-connected open surfaces with different geometry. Figure \ref{fig:foot} shows the disk conformal parameterization of the foot model in Figure \ref{fig:foot_disk_unnormalized} using our proposed method. Figure \ref{fig:statue} and \ref{fig:bunny} respectively show a simply-connected open statue model and a Stanford bunny model, together with the disk conformal parameterizations of them obtained by our proposed method. For a better visualization of the parameterization results, the triangular faces of the meshes are colored with the mean curvatures of the models. The disk parameterizations well resemble the original models locally. The histograms of the angular distortion of our proposed algorithm are shown in Figure \ref{fig:histogram}. It can be observed from the peaks of the histograms that the angle differences highly concentrate at 0. This indicates that our proposed algorithm produces disk conformal parameterizations with minimal distortions for different kinds of simply-connected open surfaces. As for the efficiency, our proposed algorithm takes only around 1 second for moderate meshes. The computation can also be complete within half a minute for dense meshes.

To quantitatively assess the performance of our proposed algorithm for disk conformal parameterizations, we compare our proposed algorithm with the existing methods. In particular, the comparisons performed in \cite{Choi15b} suggest that the holomorphic 1-form method \cite{Gu02} and the two-step iterative approach \cite{Choi15b} achieve the best conformality among all state-of-the-art approaches as well as the bijectivity. Therefore, the holomorphic 1-form method and the two-step iterative approach are considered in our comparisons. The implementation of the holomorphic 1-form method is included in the RiemannMapper Toolkit \cite{riemannmapper} written in C++ and the two-step iterative approach is implemented in MATLAB. The error thresholds in both the holomorphic 1-form method and the two-step iterative approach are set to be $\epsilon=10^{-5}$.

\begin{table*}
    \begin{center}
    \begin{tabular}{ |C{15mm}|C{10mm}|C{30mm}|C{30mm}|C{30mm}| }
    \hline
    Surfaces & No. of faces  & Holomorphic 1-form \cite{Gu02} & Two-step iteration \cite{Choi15b} & Our proposed method\\
    \cline{3-5}
    && \multicolumn{3}{ c| }{Time (s) / Mean($|distortion|$) (degrees) / SD($|distortion|$) (degrees)}\\ \hline
    Horse & 9K & fail & 0.76 / 4.58 / 5.45 & 0.18 / 4.60 / 5.49 \\ \hline
    T-shirt & 14K & 18.64 / 1.38 / 3.26 & 2.09 / 1.34 / 3.25 & 0.34 / 1.35 / 3.26 \\ \hline
    Foot & 20K & 11.73 / 1.40 / 1.24 & 1.75 / 1.42 / 1.22 & 0.47 / 1.42 / 1.22\\ \hline
    Chinese lion & 30K & 29.87 / 1.44 / 2.05 & 2.70 / 1.42 / 2.04 & 0.92 / 1.42 / 2.05\\ \hline
    Sophie & 40K & 28.29 / 0.36 / 0.61 & 5.87 / 0.34 / 0.60 & 1.31 / 0.35 / 0.60\\ \hline
    Bimba & 50K & 28.04 / 1.29 / 1.78 & 2.22 / 1.22 / 1.74 & 1.32 / 1.22 / 1.75\\ \hline
    Human face & 50K & 28.45 / 0.55 / 1.84 & 4.61 / 0.53 / 1.82 & 1.40 / 0.53 / 1.83\\ \hline
    Niccolo da Uzzano & 50K & 29.49 / 0.78 / 1.73 & 7.95 / 0.75 / 1.75 & 1.34 / 0.76 / 1.74\\ \hline
    Mask & 60K & fail & 7.08 / 0.25 / 0.33 & 1.93 / 0.25 / 0.33 \\ \hline
    Bunny & 70K & 40.14 / 1.08 / 1.80  & 4.18 / 1.08 / 1.79 & 1.99 / 1.08 / 1.79 \\ \hline
    Brain & 100K & 58.49 / 1.46 / 1.59 & 6.81 / 1.46 / 1.59 & 2.73 / 1.46 / 1.59 \\ \hline
    Lion vase & 100K & 93.64 / 1.44 / 1.91 & 5.34 / 1.27 / 1.75 & 2.64 / 1.27 / 1.76\\ \hline
    Max Planck & 100K & 75.92 / 0.61 / 0.80 & 6.58 / 0.61 / 0.80 & 2.88 / 0.61 / 0.80\\ \hline
    Hand & 110K & 63.90 / 1.40 / 1.99 & 7.51 / 1.21 / 1.31 & 3.30 / 1.21 / 1.31 \\ \hline
    Igea & 270K & 173.60 / 0.40 / 0.72 & 54.63 / 0.40 / 0.71 & 9.47 / 0.40 / 0.71 \\ \hline
    Julius Caesar & 430K & 295.63 / 0.21 / 0.67 & 65.54 / 0.20 / 0.67 & 19.51 / 0.20 / 0.67\\ \hline
    \end{tabular}
    \end{center}
    \caption{Performance of our proposed method and two other state-of-the-art algorithms. Here the distortion refers to the angular distortion of the parameterization, that is, the difference (in degrees) between the new angle and its corresponding original angle on the mesh.}
    \label{performance}
\end{table*}

The statistics of the performances of the algorithms are shown in Table \ref{performance}. It is noteworthy that the angular distortion of our proposed method is comparable and sometimes better than the two state-of-the-art algorithms \cite{Gu02,Choi15b}. This implies that our algorithm successfully produces low-distortion parameterizations. Moreover, our proposed method demonstrates a significant improvement in the computational time of the disk conformal parameterizations. Specifically, our proposed method is around 20 times faster than the holomorphic 1-form method on average. Also, our proposed method requires 60\% less computational time than the two-step iteration on average. The results illustrate the efficiency of our proposed algorithm. As a remark, in all experiments, our proposed algorithm generates folding-free parameterization results.

\begin{figure}[t]
\begin{center}
   \includegraphics[width=0.95\linewidth]{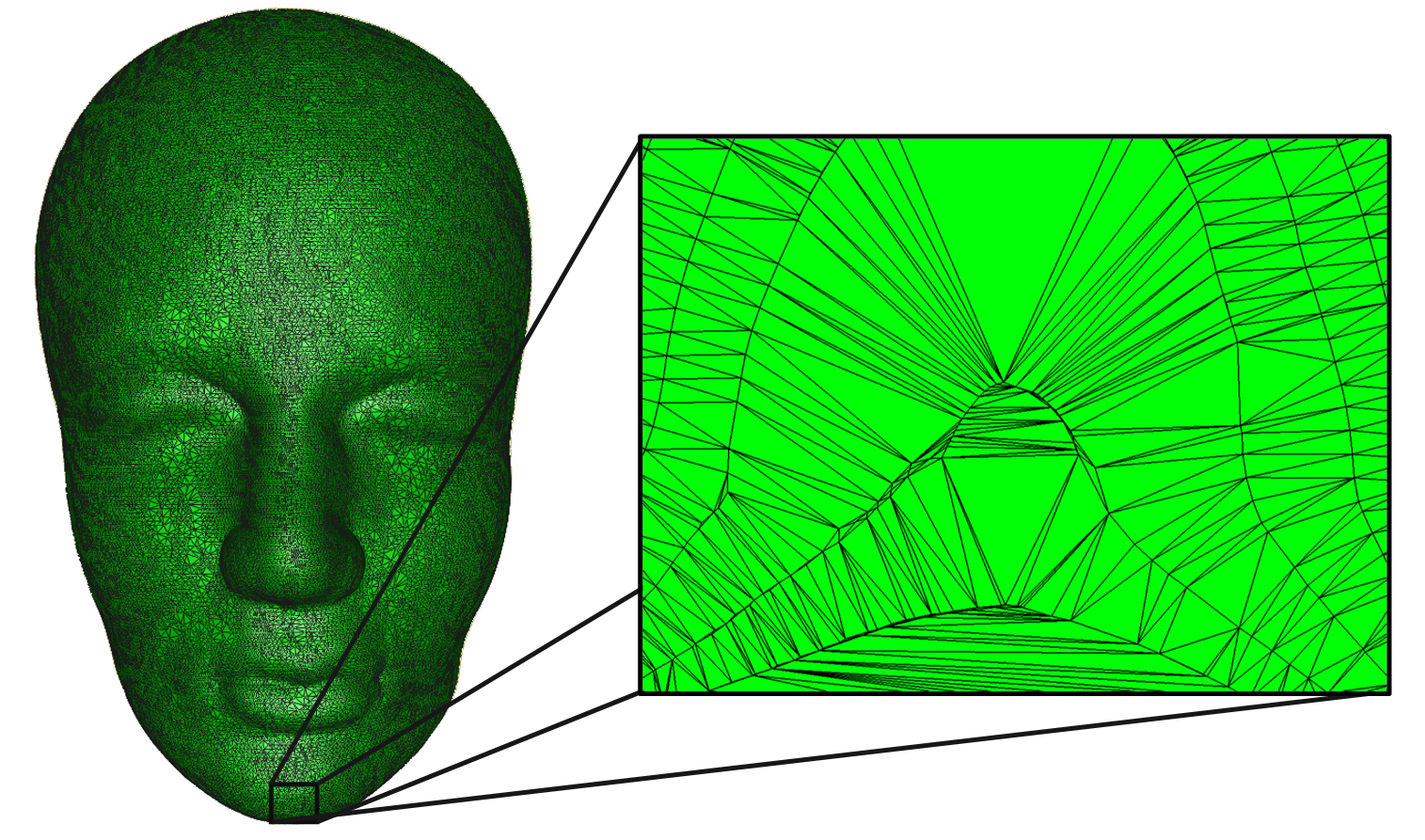}

\end{center}
   \caption{A mask model with a highly irregular triangulation. Sharp triangles can be easily observed in a zoom-in of the model. Our proposed method can handle this kind of irregular triangulations. Left: The mask model. Right: A zoom-in of the model.}
\label{fig:mask}
\end{figure}

Our proposed method is very robust to irregular triangulations. Figure \ref{fig:mask} shows a simply-connected mask model which has a very irregular triangulation. Sharp and irregular triangles can be easily observed in a zoom-in of the model. Note that the holomorphic 1-form method \cite{Gu02} fails to compute the disk conformal parameterization of this model while our proposed method succeeds (please refer to Table \ref{performance}). This demonstrates the robustness of our proposed method.

\begin{figure}[t]
\begin{center}
   \includegraphics[width=0.5\linewidth]{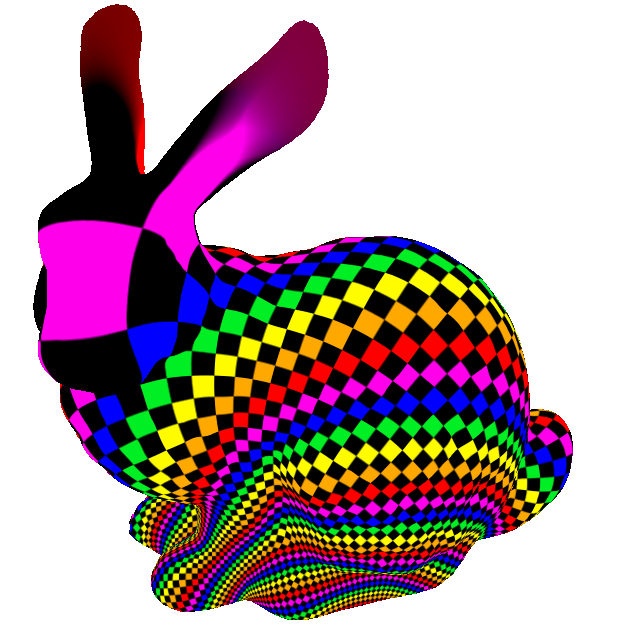}
\end{center}
   \caption{The Stanford bunny model with a rainbow checkerboard texture mapped onto it using our proposed parameterization algorithm.}
\label{fig:texture}
\end{figure}

\begin{figure}[t]
\begin{center}
   \includegraphics[width=0.5\linewidth]{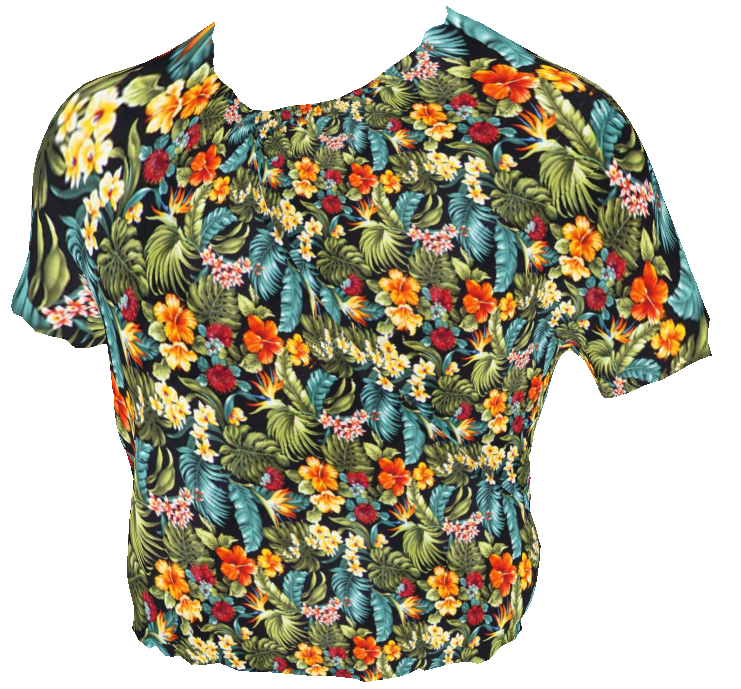}
\end{center}
   \caption{A T-shirt model with a flower pattern design mapped onto it using our proposed parameterization algorithm.}
\label{fig:cloth}
\end{figure}

In addition, to illustrate the accuracy of our proposed algorithm, we map a rainbow checkerboard texture onto a simply-connected bunny mesh using our proposed parameterization algorithm. In Figure \ref{fig:texture}, it can be easily observed that the orthogonal checkerboard structure is well preserved under the our proposed parameterization scheme. The preservation of the orthogonality indicates that our proposed algorithm is with negligibly low conformality distortion.

Because of the efficiency, accuracy, bijectivity and robustness, our proposed algorithm is highly practical in applications, such as texture mapping in computer graphics and fashion design. For instance, clothing designs can be easily visualized using our proposed algorithm. Figure \ref{fig:cloth} shows a T-shirt model with a 2D flower pattern design mapped onto it. It is noteworthy that the patterns are well preserved on the 3D T-shirt model. Unlike the traditional 2D images of clothes, the 3D models with our texture mapping technique provide a more comprehensive and realistic preview of the clothes for the designers. In addition, our proposed algorithm can be used in developing virtual dressing rooms. In virtual dressing rooms, it is desirable to have an efficient and accurate way for customers to virtually try on clothes. Our proposed algorithm can help creating such a virtual dressing experience for online shoppers.

\begin{table*}
    \begin{center}
    \begin{tabular}{ |C{20mm}|C{10mm}|C{45mm}|C{45mm}| }
    \hline
    Surfaces & No. of faces  & Our proposed method & Our proposed method without the South-pole step in \cite{Choi15} \\ 
    \cline{3-4}
    && \multicolumn{2}{ c| }{Time (s) / Mean($|distortion|$) (degrees) / SD($|distortion|$) (degrees)}\\ \hline
    Horse & 9K & 0.18 / 4.60 / 5.49 & 0.11 / 4.63 / 5.50\\ \hline 
    T-shirt & 14K & 0.34 / 1.35 / 3.26 & 0.22 / 1.37 / 3.27\\ \hline 
    Foot & 20K & 0.47 / 1.42 / 1.22 & 0.31 / 1.43 / 1.23\\ \hline 
    Chinese lion & 30K & 0.92 / 1.42 / 2.05 & 0.56 / 1.43 / 2.05\\ \hline 
    Sophie & 40K & 1.31 / 0.35 / 0.60 & 0.81 / 0.37 / 0.62\\ \hline 
    Bimba & 50K & 1.32 / 1.22 / 1.75 & 0.75 / 1.23 / 1.76\\ \hline
    Human face & 50K & 1.40 / 0.53 / 1.83 & 0.83 / 0.55 / 1.84\\ \hline
    Niccolo da Uzzano & 50K & 1.34 / 0.76 / 1.74 & 0.78 / 0.77 / 1.74\\ \hline
    Mask & 60K & 1.93 / 0.25 / 0.33 & 1.22 / 0.28 / 0.37\\ \hline
    Bunny & 70K & 1.99 / 1.08 / 1.79 & 1.14 / 1.09 / 1.79\\ \hline
    Brain & 100K & 2.73 / 1.46 / 1.59 & 1.55 / 1.49 / 1.60\\ \hline 
    Lion vase & 100K & 2.64 / 1.27 / 1.76 & 1.45 / 1.28 / 1.76\\ \hline
    Max Planck & 100K & 2.88 / 0.61 / 0.80 & 1.54 / 0.63 / 0.80\\ \hline 
    Hand & 110K & 3.30 / 1.21 / 1.31 & 1.74 / 1.22 / 1.31\\ \hline
    Igea & 270K & 9.47 / 0.40 / 0.71 & 5.22 / 0.54 / 0.74\\ \hline
    Julius Caesar & 430K & 19.51 / 0.20 / 0.67 & 12.41 / 0.21 / 0.68 \\ \hline
    \end{tabular}
    \end{center}
    \caption{Performance of the current version of our proposed method and a possible improved version of it without the South-pole step in \cite{Choi15}, provided a suitable choice of the boundary triangle in Equation (\ref{eqt:laplace}).}
    \label{enhancement}
\end{table*}

As a final remark, we discuss the possibility of further boosting up the efficiency of our proposed algorithm. Recall that in our proposed algorithm, we need to compute the spherical conformal parameterization of the double covered surface using the algorithm in \cite{Choi15} in order to find an initial planar parameterization. The algorithm in \cite{Choi15} consists of two major steps, namely a North-pole step and a South-pole step. In the North-pole step, the Laplace equation (\ref{eqt:laplace}) is solved in a triangular domain $[b_1,b_2,b_3]$ on the complex plane and the inverse stereographic projection is applied. Then the South-pole step aims to correct the conformality distortion near the North pole of the sphere caused by the discretization and approximation errors. In fact, since we are only interested in half of the glued surface, the South-pole step may be skipped as we can take the Southern hemisphere obtained by the first step as our result. It may already be with acceptable conformality.

The conformality distortion in the North-pole step in \cite{Choi15} is primarily caused by the choice of the boundary triangle $[a_1,a_2,a_3]$. If the chosen boundary triangle and its neighboring triangular faces are regular enough, then the overall distortion obtained in the North-pole step is already very low and only a small region around the North-pole is with relatively large distortion. In this case, the South-pole step is not needed to improve the distortion of the Southern region which we are interested. On the other hand, with the presence of the South-pole step, the overall distortion in the final spherical parameterization result will be negligible regardless of the choice of the boundary triangle. In short, with a well chosen boundary triangle $[a_1,a_2,a_3]$ in Equation (\ref{eqt:laplace}), half of the computational cost in computing the spherical conformal mapping can be further reduced.

Table \ref{enhancement} shows the performance of the current version of our proposed method and the possible improved version of it without the South-pole step in \cite{Choi15}, under a suitable choice of the boundary triangle $[a_1,a_2,a_3]$ in Equation (\ref{eqt:laplace}). It can be noted that the differences in the conformality distortion of the two versions are mostly negligible, while the version without the South-pole step possess a further 40\% improvement in the computational time on average. However, it should be reminded that the possible improved version without the South-pole step requires a suitably chosen boundary triangle $[a_1,a_2,a_3]$ for \cite{Choi15}, while the current version of our method is fully automatic. Hence, the current version of our proposed method is probably more suitable for practical applications until an automatic algorithm for searching for the most suitable boundary triangle $[a_1,a_2,a_3]$ is developed.

\section{Conclusion and future works} \label{conclusion}
In this paper, we have proposed a linear formulation for the disk conformal parameterizations of simply-connected open surfaces. We begin the algorithm by obtaining an initial planar parameterization via double covering and spherical conformal mapping. Note that even the size of the surface is doubled by double covering, the combination of the double covering technique and the spherical conformal mapping results in an efficient computation because of the symmetry. After that, we normalize the boundary and compose the map with a quasi-conformal map so as to correct the conformality distortion and achieve the bijectivity. Our proposed formulation is entirely linear, and hence the computation is significantly accelerated by over 60\% when compared with the fastest state-of-the-art approaches. At the same time, our parameterization results are of comparable quality to those produced by the other state-of-the-art approaches in terms of the conformality distortions, the bijectivity and the robustness. Therefore, our proposed algorithm is highly practical in real applications, especially for the problems for which the computational complexity is the main concern. In the future, we plan to explore more applications, such as remeshing and registration of simply-connected open surfaces, based on the proposed parameterization scheme.


\begin{thebibliography}{}
\bibitem{Angenent99}
S. Angenent, S. Haker, A. Tannenbaum, and R. Kikinis, {\em Conformal geometry and brain flattening}. Medical image computing and computer-assisted intervention (MICCAI), 271--278, 1999.
  
\bibitem{Choi15}
P.T. Choi, K.C. Lam, and L.M. Lui, {\em FLASH: Fast Landmark Aligned Spherical Harmonic Parameterization for Genus-0 Closed Brain Surfaces}, SIAM Journal on Imaging Sciences, Volume 8, Issue 1, pp. 67--94, 2015.

\bibitem{Choi15b}
P.T. Choi and L.M. Lui, {\em Fast Disk Conformal Parameterization of Simply-connected Open Surfaces}, Journal of Scientific Computing, Volume. 65, Issue 3, pp. 1065--1090, 2015. 

\bibitem{Desbrun02}
M. Desbrun, M. Meyer, and P. Alliez, {\em Intrinsic Parameterizations of Surface Meshes}, Computer Graphics Forum, Volume 21, pp. 209--218, 2002.

\bibitem{Docarmo76}
M. do Carmo, {\em Differential Geometry of Curves and Surfaces}, Prentice Hall, 1976.

\bibitem{Floater97}
M. Floater, {\em Parameterization and Smooth Approximation of Surface Triangulations}, Computer Aided Geometric Design, Volume 14, Issue 3, pp. 231--250, 1997.

\bibitem{Floater03}
M. Floater, {\em Mean Value Coordinates}, Computer Aided Geometric Design, Volume 20, Issue 1, pp. 19--27, 2003.

\bibitem{Floater02}
M. Floater and K. Hormann, {\em Parameterization of Triangulations and Unorganized Points}, Tutorials on Multiresolution in Geometric Modelling, pp. 287--316, 2002.

\bibitem{Floater05}
M. Floater and K. Hormann, {\em Surface Parameterization: A Tutorial and Survey}, Advances in Multiresolution for Geometric Modelling, pp. 157--186, 2005.

\bibitem{Gardiner00}
F. Gardiner and N. Lakic, {\em Quasiconformal Teichm\"{u}ller Theory}, American Mathematics Society, 2000.

\bibitem{Gotsman03}
C. Gotsman, X. Gu, and A. Sheffer, {\em Fundamentals of spherical parameterization for 3D meshes}. ACM Transactions on Graphics (Proceedings of ACM SIGGRAPH 2003), 22(3):358--363, 2003.

\bibitem{Gu02}
X. Gu and S.T. Yau, {\em Computing Conformal Structures of Surfaces}, Communications in Information and Systems, Volume 2, Issue 2, pp. 121--146, 2002.

\bibitem{Gu03}
X. Gu and S.T. Yau, {\em Global Conformal Surface Parameterization}, Eurographics Symposium on Geometry Processing, pp. 127--137, 2003.

\bibitem{Gu04}
X. Gu, Y. Wang, T. F. Chan, P. M. Thompson, and S.-T. Yau, {\em Genus Zero Surface Conformal Mapping and its Application to Brain Surface Mapping}, IEEE Transactions on Medical Imaging, Volume 23, pp. 949--958, 2004.

\bibitem{Haker00}
S. Haker, S. Angenent, A. Tannenbaum, R. Kikinis, and G. Sapiro, {\em Conformal Surface Parameterization for Texture Mapping}, IEEE Transactions on Visualization and Computer Graphics, Volume 6, Issue 2, pp. 181--189, 2000.

\bibitem{Hormann00}
K. Hormann and G. Greiner, {\em MIPS: An Efficient Global Parametrization Method}, Curve and Surface Design, pp. 153--162, 2000.

\bibitem{Hormann01}
K. Hormann, U. Labsik, and G. Greiner, {\em Remeshing Triangulated Surfaces with Optimal Parameterizations}, Computer-Aided Design, Volume 33, Issue 11, pp. 779--788, 2001.

\bibitem{Hormann07}
K. Hormann, B. L\'evy, and A. Sheffer, {\em Mesh Parameterization: Theory and Practice}, SIGGRAPH 2007 Course Notes, Volume 2, pp. 1--122, 2007.

\bibitem{Jin05}
M. Jin, Y. Wang, X. Gu, and S.T. Yau, {\em Optimal Global Conformal Surface Parameterization for Visualization}, Communications in Information and Systems, Volume 4, Issue 2, pp. 117--134, 2005.

\bibitem{Jin08}
M. Jin, J. Kim, F. Luo, and X. Gu, {\em Discrete Surface Ricci Flow}, IEEE Transactions on Visualization and Computer Graphics, Volume 14, Issue 5, pp. 1030--1043, 2008.

\bibitem{Jost11}
J. Jost, {\em Riemannian Geometry and Geometric Analysis}. Universitext, Springer, 2011.

\bibitem{Kharevych05}
L. Kharevych, B. Springborn, and P. Schr\"oder, {\em Discrete Conformal Mappings via Circle Patterns}, ACM Transactions on Graphics, Volume 25, Issue 2, pp. 412--438, 2006.

\bibitem{Levy02}
B. L\'evy, S. Petitjean, N. Ray, and J. Maillot, {\em Least Squares Conformal Maps for Automatic Texture Atlas Generation}, ACM Transactions on Graphics (Proceedings of ACM SIGGRAPH 2002), pp. 362--371, 2002.

\bibitem{Lui10}
L. M. Lui, S. Thiruvenkadam, Y. Wang, P. Thompson, and T. F. Chan, {\em Optimized Conformal Surface Registration with Shape-based Landmark Matching}, SIAM Journal on Imaging Sciences, Volume 3, Issue 1, pp. 52--78, 2010.

\bibitem{Lui13}
L.M. Lui, K.C. Lam, T.W. Wong, and X. Gu, {\em Texture Map and Video Compression Using Beltrami Representation}, SIAM Journal on Imaging Sciences, Volume 6, Issue 4, pp. 1880--1902, 2013.

\bibitem{Luo04}
F. Luo, {\em Combinatorial Yamabe Flow on Surfaces}, Communications in Contemporary Mathematics, Volume 6, pp. 765--780, 2004.

\bibitem{Mullen08}
P. Mullen, Y. Tong, P. Alliez and M. Desbrun, {\em Spectral Conformal Parameterization}, Computer Graphics Forum, Volume 27, Issue 5, pp. 1487--1494, 2008.

\bibitem{Pinkall93}
U. Pinkall and K. Polthier, {\em Computing Discrete Minimal Surfaces and their Conjugates}, Experimental Mathematics, Volume 2, Issue 1, pp. 15--36, 1993.

\bibitem{Praun03}
E. Praun and H. Hoppe, {\em Spherical Parametrization and Remeshing}, ACM Transactions on Graphics (Proceedings of ACM SIGGRAPH 2003), Volume 22, Issue 3, pp. 340--349, 2003.

\bibitem{Remacle10}
J-F Remacle, C. Geuzaine, G. Comp\`{e}re, and E. Marchandise, {\em High Quality Surface Remeshing Using Harmonic Maps}, International Journal for Numerical Methods in Engineering, Volume 83, Issue 4, pp. 403--425, 2010.

\bibitem{Schoen94}
R. Schoen and S.T. Yau, {\em Lectures on Differential Geometry}, International Press, 1994.

\bibitem{Schoen97}
R. Schoen and S. Yau, {\em Lectures on Harmonic Maps}, Cambridge, MA:Harvard Univ., Int. Press, 1997.

\bibitem{Sheffer00}
A. Sheffer and E. de Sturler, {\em Surface Parameterization for Meshing by Triangulation Flattening}, Proceedings of 9th International Meshing Roundtable, pp. 161--172, 2000.

\bibitem{Sheffer05}
A. Sheffer, B. L\'evy, M. Mogilnitsky, and A. Bogomyakov, {\em ABF++: Fast and Robust Angle Based Flattening}, ACM Transactions on Graphics, Volume 24, Issue 2, pp. 311--330, 2005.

\bibitem{Sheffer06}
A. Sheffer, E. Praun and K. Rose, {\em Mesh Parameterization Methods and their Applications}, Foundations and Trends in Computer Graphics and Vision, Volume 2, Issue 2, pp. 105--171, 2006.

\bibitem{Yang09}
Y.L. Yang, R. Guo, F. Luo, S.M. Hu, and X.F. Gu, {\em Generalized Discrete Ricci Flow}, Computer Graphics Forum, Volume 28, Issue 7, pp. 2005--2014, 2009.

\bibitem{Zhang05}
E. Zhang, K. Mischaikow, and G. Turk, {\em Feature-based Surface Parameterization and Texture Mapping}, ACM Transactions on Graphics, Volume 24, Issue 1, pp. 1--27, 2005.

\bibitem{aim@shape}
{\em AIM@SHAPE Shape Repository}. http://shapes.aim-at-shape.net/

\bibitem{stanford}
{\em Stanford 3D Scanning Repository}. http://graphics.stanford.edu/data/3Dscanrep/

\bibitem{princeton}
{\em A Benchmark for 3D Mesh Segmentation}. http://segeval.cs.princeton.edu/

\bibitem{riemannmapper}
{\em RiemannMapper: A Mesh Parameterization Toolkit}.\\http://www3.cs.stonybrook.edu/\~{}gu/software/RiemannMapper/

\end{thebibliography}
\end{document}